\documentclass[twocolumn, 10 pt]{IEEEtran}
\usepackage{times,epsfig,graphicx,wrapfig,bbm,color,amssymb}
\usepackage{amsmath}

\newcommand {\cA}{{\mathcal{A}}}

\newcommand {\cD}{{\mathcal{D}}}

\newcommand {\cX}{{\mathcal{X}}}

\newcommand {\bff} {{\bf f}}
\newcommand {\ba} {{\bf a}}

\newcommand {\bx} {{\bf x}}

\newcommand {\bs} {{\bf s}}
\newcommand {\bw} {{\bf w}}

\newcommand{\avg}{{\rm avg}}

\newcommand {\N} {{\rm I\kern-1.5pt N}}
\newcommand {\R} {{\rm I\kern-2.5pt R}}

\newtheorem{lemma}{Lemma}

\newtheorem{coro}{Corollary}
\newtheorem{theorem}{Theorem}
\newtheorem{defn}{Definition}
\newtheorem{assm}{Assumption}

\newcommand{\beqa}{\begin{eqnarray}}
\newcommand{\eeqa}{\end{eqnarray}}
\newcommand{\beqan}{\begin{eqnarray*}}
\newcommand{\eeqan}{\end{eqnarray*}}
\newcommand{\beq}{\begin{equation}}
\newcommand{\eeq}{\end{equation}}

\newcommand{\bfl}{\begin{flushleft}}
\newcommand{\efl}{\end{flushleft}}

\newcommand{\myb}{\hspace{-0.1in}}

\newcommand{\myeq}{& \hspace{-0.1in} = & \hspace{-0.1in}}

\newcommand{\lb}{\nonumber \\}
\newcommand{\myarr}{\begin{array}{lll}}

\newcommand{\mygeq}{& \myb \geq & \myb}

\newcommand{\myleq}{& \myb \leq & \myb}

\newcommand{\bitem}{\begin{itemize}}
\newcommand{\eitem}{\end{itemize}}
\newcommand{\benum}{\begin{enumerate}}
\newcommand{\eenum}{\end{enumerate}}

\newcommand{\myhb}{\hspace{-0.3in}}

\newcommand{\opt}{{\rm opt}}

\def\QED{~\rule[-1pt]{5pt}{5pt}\par\medskip}
\newenvironment{proof}{{\bf Proof: \ }}{ \hfill \QED}

\newcommand{\myskip}{\\ \vspace{-0.1in}}

\begin{document}

\title{Internalization of Externalities in 
	Interdependent Security: Large Network Cases}

\author{Richard J. La\thanks{This work was 
supported in part by contracts 70NANB13H012 and 
70NANB14H015 from National Institute
of Standards and Technology.} 
\thanks{Author is with the Department of Electrical \& 
Computer Engineering (ECE) and the Institute for Systems 
Research (ISR) at the University of Maryland, College Park.
E-mail: hyongla@umd.edu}
}

\maketitle

\begin{abstract}
With increasing connectivity among comprising agents or  
(sub-)systems in large, 
complex systems,  there is a growing interest in understanding 
interdependent security and dealing with inefficiency in security
investments. 
Making use of a population game model and the well-known 
Chung-Lu random graph model, 
we study how one could encourage selfish agents to invest more in 
security by internalizing the externalities produced by their
security investments. 

To this end, we first establish an interesting relation between 
the local minimizers of social cost and the Nash equilibria of 
a population game with slightly altered costs. Secondly, under
a mild technical assumption, we demonstrate that there exists 
a unique minimizer of social cost and it 
coincides with the unique Nash equilibrium
of the population game.  This finding tells us how to modify 
the private cost functions of selfish agents in order to enhance
the overall security and reduce social cost.  In addition, it reveals
how the sensitivity of overall security to security investments
of agents influences their externalities and, consequently, 
penalties or taxes that should be imposed for internalization of
externalities.  
Finally, we illustrate how the degree distribution of agents
influences their security investments and overall security
at both the NEs of population games and social optima.
\end{abstract}


\IEEEpeerreviewmaketitle

\section{Introduction}	\label{sec:Introduction}

Today, many engineering, financial and social systems (e.g., the 
Internet, power grids, equity and commodity markets, and social 
networks) are highly connected. 
In such systems, the security of an individual, organization 
or system is often 
dependent not only on its own security measures, but also on those 
of others.\footnote{We refer to individuals, organizations 
and even countries in these settings as {\em agents}.}
This is dubbed {\em interdependent security} (IDS) by Heal and 
Kunreuther \cite{HealKun2004}, and 
arises naturally in many areas including 
cybersecurity \cite{BolotLelarge2008, Jiang2011,
LelargeBolot2008, LelargeBolot2009}, 
cyber-physical systems security (e.g., power grids)
\cite{NIST_SmartGrid, Bou-Harb}, 
epidemiology \cite{Pastor2005, Schneider2011}, 
financial networks and systems~\cite{Beale2011, Caccioli2011, 
Caccioli2012}, 
homeland security~\cite{HealKun2002, KunMichel2009}, and
supply chain and transportation system security (e.g., 
airline security)~\cite{Gkonis2010, HealKun2004, 
KearnsOrtiz, KunHeal2003}. Given 
the increasing connectivity among the systems
comprising critical infrastructures (e.g., smart grids), IDS has 
emerged as an active and vital research area. 

The coupling or interdependence in security among agents in IDS 
is often modeled using a {\em dependence graph}: the nodes are 
agents, and the existence of an (undirected) edge between 
two nodes signals interdependence of their security. 
Some of key challenges to tackling IDS in {\em large} 
networks or systems
are: (i) participating agents are often strategic and are 
interested only in their own security or objectives, 
rather than the security or cost of the overall system, 
(ii) when an agent invests in security measures, it produces
{\em externalities} and {\em network effects}
for its neighbors, and 
(iii) any attempt to model {\em detailed} interactions 
among many strategic agents suffers from the {\em curse 
of dimensionality}.

Let us illustrate some of these concepts and motivations for
our study using following examples.
\\ \vspace{-0.12in}

\noindent {\bf {\em E1.}}
{\em Spread of malware through emails:} 
When a user is infected by malware, it can scan the user's 
emails or the hard disk drive of the infected machine 
and send the 
user's personal or other confidential information to 
criminals interested in stealing, for instance, the user's 
identify (ID) or trade secrets. 
Moreover, the malware can browse the user's address book 
and either forward it to attackers or 
send out bogus emails, i.e., email spoofing, 
with a link or an attachment to those on the contact list. 
When a recipient clicks on the link or
opens the attachment, it too becomes infected.

In order to reduce the risks or threats from
malware, users can 
install an anti-malware utility on their 
devices. When a user adopts an anti-malware tool, 
not only does it reduce its own risk, but it
also curbs the risk to those 
on its address book for the reason stated above, 
thereby protecting its friends to some degree. 
Therefore, it produces {\em positive externalities} 
for others \cite{ShapiroVarian, Varian_Microeconomics}. 
Interestingly, these positive externalities diminish 
the value of installing
anti-malware utilities for others, thus introducing
{\em negative network effects} for them. 
\\ \vspace{-0.12in}

\noindent {\bf {\em E2.}}
{\em Organizational networks:} 
Organizational information networks are typically 
protected by various security measures, including 
replication of data storage and information, network 
monitoring systems and incoming traffic monitoring. The
choices of employed security measures may depend on the 
magnitude of potential financial and other losses (e.g., 
damages to reputation), 
desired network/system dependability as well as
available budget for security investments~\cite{AndersonMoore}. 

Organizational networks are interconnected and, in many cases, 
share information. Thus, when some networks are more vulnerable, 
they may serve as a stepping stone for sophisticated hackers to 
gain a foothold inside the network and ultimately 
access to even better protected high-value targets by using 
known vulnerabilities \cite{NVD} or zero-day exploits
\cite{BilgeDumit2012}. 
\\ \vspace{-0.1in}

It is well documented \cite{Bary, Gordon2015, Varian} 
that the selfish nature of agents leads 
to inefficiency in many settings, including
under-investments in security, 
thanks to {\em free riding} 
in part caused by positive externalities
illustrated in the first example. 
Therefore, researchers have been searching for ways to
improve overall security, including cyberinsurance and 
incentive mechanisms to increase the security investments 
by selfish agents (e.g., \cite{Hofmann, Jiang2011, LelargeBolot2009, 
Ogut2005, Zhao2009}). 

One promising approach to enhancing overall security is through 
{\em internalization of externalities}~\cite{Varian_Microeconomics},
which is the focus of this paper. 
This requires altering the (private) cost of agents so that their 
costs reflect their contribution to the social cost. A key challenge 
to implementing this 
lies in correctly quantifying the externalities
produced by each agent and accounting for them in its cost function.
 
We aim to explore (a) how we can measure (or approximate)
the externalities generated by agents and signal them correctly
to the agents, and 
(b) how the {\em sensitivity} of overall security to the security 
investments of agents shapes the externalities and, hence, 
penalties/taxes that ought to be levied on them (as a 
part of their cost functions) to 
internalize their externalities. 
In particular, we are interested in scenarios with {\em many} 
agents in large networks and systems. 

To this end, we consider scenarios where malicious entities
or attackers 
launch attacks against agents, for example, in hopes 
of infecting/taking control of machines or 
gaining unauthorized access to private
information of victims. Not only can agents suffer damages or
losses from {\em direct} attacks by attackers, but the 
victims of successful direct attacks may also unknowingly 
help the attackers unleash {\em indirect} attacks on their 
neighbors. 

We assume that the agents are rational and interested in 
minimizing their own costs. In the face of possible attacks, 
they manage their risks by investing in a number of 
security measures. Unfortunately, it is hard to 
model the details of strategic interactions among 
many agents due to the curse of dimensionality.

In order to skirt this difficulty, we employ a {\em population 
game} model \cite{Sandholm} with the help of the so-called
{\em Chung-Lu random graph model}~\cite{ChungLu}. A population
game is often used to study strategic interactions between 
a large number of agents, possibly from different populations. 
This novel framework allows us to examine the network- or 
system-level security in IDS settings with many 
comprising agents. 

We adopt a well known solution concept, namely {\em Nash 
equilibrium} (NE) of the population game, as an approximation 
to agents' behavior in practice.   
Our goal is to understand how the NEs of population games 
are related to the social optima with the objective of 
identifying a potential means of internalizing the 
externalities produced by agents' security investments. 

Our main contributions can be summarized as follows. 
\benum

\item We reveal an intriguing relation between the NEs of 
population games and social optima. More specifically, we
show that {\em any local minimizer} of social cost is
a pure-strategy 
NE of a slightly modified population game in which agents'
cost functions are altered to account for the externalities
brought on by their security decisions. Therefore, the set of
local minimizers of social cost is {\em contained}
in the set of pure-strategy NEs of the modified population 
game (Section \ref{subsec:Relation}).

\item Using this relation, we establish 
under a mild technical condition that there exists a 
{\em unique} pure-strategy NE of the aforementioned modified 
population game and it coincides with the {\em unique} 
(global) minimizer of social cost,
without requiring convexity of social cost 
(Section~\ref{subsec:Internalization}).

\item We demonstrate that an agent 
with a {\em fixed} degree suffers {\em fewer} attacks 
both at an NE of the population game and at a social 
optimum as the weighted node 
degree distribution of the dependence graph 
(with node degrees as weights) 
becomes stochastically larger \cite{ShakedShan}. Hence, 
as the dependence graph 
becomes more connected, indicating a higher level of
interdependence in security, a fixed-degree node
will likely experience better local security 
(Section~\ref{sec:Property}). 

\eenum

To the best of our knowledge, our work is the first study 
to examine the possibility of internalizing externalities,
based on the relation between the NEs of noncooperative
games and social optima with many agents in IDS settings. 
It is true that our study is conducted using a simplified 
model under several assumptions for 
analytical tractability. For this reason, it is not 
our intention to claim 
that our model accurately represents the complex reality 
and the {\em quantitative} aspects of our findings will 
hold in practice. 

Instead, our hope is that even this simple model will 
help us develop valuable insight into the {\em qualitative} 
nature of (aggregate) behavior of agents with help of 
analytical findings. 
Furthermore, our findings will likely shed some 
light on (a) how we can improve the overall security 
via internalization of externalities and (b) how the 
underlying interdependency of security among agents affects the
security experienced by them and in turn influences their 
security investments in more realistic settings.  

The rest of the paper is organized as follows. We briefly 
summarize some
of existing studies that are most closely related to 
our study in Section~\ref{sec:Related}. Section \ref{sec:Model}
describes the population game model we adopt for our analysis, 
followed by some preliminary results in Section~\ref{sec:Preliminary}. 
Our main findings on the relation between the NEs of population 
games and the local minimizers of social cost 
as well as the approximation 
of (negative) externalities produced by agents are presented 
in Section~\ref{sec:Internalization}. The effects of weighted
node degree distribution on local security and negative
externalities are reported in Section~\ref{sec:Property}. 
Some numerical results are presented in Section
\ref{sec:Numerical}. 
We conclude in Section~\ref{sec:Conclusion}.

\section{Related literature} 	\label{sec:Related}

As mentioned in Section~\ref{sec:Introduction}, 
Kunreuther and Heal \cite{HealKun2004, KunHeal2003}
studied IDS where the security of involved 
parties is interdependent. IDS scenarios 
with strategic agents are often studied in game theoretic 
settings, e.g., \cite{Gross2008, Jiang2011, 
KearnsOrtiz, Miura2008a}. 
We refer an interested reader to a survey paper by Laszka et al. 
\cite{Laszka} and references therein for a succinct discussion of 
these and other related studies. Here, we focus on several 
studies that are most relevant to ours and summarize their main 
findings. 

It is well known that the existence of externalities
often leads to inefficient equilibria due to market 
failures (e.g., \cite{Gordon2015, Zhao2009}). 
In particular, it is shown that positive externalities 
(resp. negative externalities) produced by security 
investments result in under-investments 
(resp. over-investments) in security~\cite{Zhao2009}.

In order to address this inefficiency in security investments, 
researchers explored
various means of internalizing externalities, including
taxation, cyberinsurance, regulations and coordinating
mechanisms~\cite{KunHeal2003} with cyberinsurance being
a popular approach extensively studied in the literature
\cite{Hofmann, LelargeBolot2009, Ogut2005, Zhao2009}.

Lelarge and Bolot \cite{LelargeBolot2009} studied the 
problem of incentivizing organizations to invest in security 
through taxation and insurance. They showed that, {\em in the 
absence of moral hazard}, insurance may be used to encourage
organizations to protect themselves. However, they did not 
suggest how the issue of moral hazard can be handled in 
practice, which is well known in the insurance
field and is difficult to rid of. 

Ogut et al.~\cite{Ogut2005} investigated the impact of 
interdependency of security on the choices for
security investments and cyberinsurance. Their findings 
show that the interdependence of security tends to reduce
the organizations' incentive to invest in security measures and 
cyberinsurance. More importantly, they suggest that even a more 
mature or developed cyberinsurance market may not 
promote cyberinsurance unless the price of insurance comes 
down and that the cyberinsurance market may fail due to 
correlated damages/incidents caused by risk interdependency, 
which can result in catastrophic losses for insurers. 

Hofmann~\cite{Hofmann} examined the possibility of monopolistic 
insurer and demonstrated that insurance monopoly
can result in higher efficiency than a competitive insurance market.
Moreover, the author suggested that the monopolistic insurer might be 
able to achieve the social optimum and reduced losses
by exercising premium discrimination
even in the case of imperfect information.  

In another interesting study, Zhao et al.~\cite{Zhao2009} considered
two alternative risk management schemes -- risk pooling 
arrangements  (RPAs) and managed security services (MSSs). 
They showed that while RPAs can be used to complement cyberinsurance
to address the over-investment issue in the case of negative 
externalities from security investments, 
it is not effective at coping with under-investments
when security investments generate positive externalities because it is not
incentive-compatible. In addition, not surprisingly, an MSS provider can 
internalize the externalities of security investments. However, 
the agents have an incentive to outsource their security management
to an MSS provider only when the number of agents is small. Thus, 
this approach fails to address the security investment inefficiency with 
a large number of agents.  

Naghizadeh and Liu~\cite{NagLiu2014} studied the problem of internalizing 
the externalities produced by strategic agents in IDS games. 
They proposed an {\em incentive-compatible} algorithm that allows the 
players to converge to a socially efficient state at an NE. 
However, the proposed scheme is not individually rational, 
and the existence of an algorithm that simultaneously achieves i)
incentive compatibility, ii) efficiency and iii) individual 
rationality remains an open problem. 

While the above studies explore different approaches to internalizing the
externalities of security investments, there are major differences from 
our study. 
First, these studies do not examine the relation between the
equilibria of IDS games and social optima. Second, our model attempts
to reflect the underlying dependence graph that captures the 
interdependence in security among agents by modeling varying degrees
of the agents,
in order to estimate their security 
risks. The studies in \cite{Hofmann, Ogut2005, Zhao2009}
do not take into account the properties of underlying dependence graph. 
Finally, we study how the degree distribution of agents (in the dependence
graph) affects their security risks and investments as well as the ensuing
externalities. To the best of our knowledge, 
this is the first analytical result that sheds light on how the 
degree distribution of dependence graph shapes the resulting network
security and externalities. 

We studied a related problem in \cite{La_TON_Cascade}, namely
how various network parameters influence the {\em (global) 
cascade} probability. We argued that the cascade 
probability can be considered a {\em global} measure of 
network security, for it measures how likely an infection, 
starting with one or a small number of initially infected 
agents, may spread to a large number of other agents. 
Not only are the emphasis and findings of \cite{La_TON_Cascade}
different from those of the current study, but also 
the model we employ in \cite{La_TON_Cascade} is different 
from that used here in two ways. First, we allowed only binary 
security choices to facilitate the analysis in~\cite{La_TON_Cascade},
whereas we consider a continuous action space representing varying
amounts of security investments in this study. 
Second, in 
\cite{La_TON_Cascade} we assumed a specific way in which 
infections can transmit multiple hops and studied how the infection 
propagation rate affects resulting cascade probabilities at the NEs. 
In the current study, we do not fix the dynamics of infection 
propagation through a network. Instead, we abstract out
security risks that agents might see under different propagation 
models by using a function and investigate how the `shape' of
the function affects the externalities of security investments and 
resulting penalties necessary to internalize them.

\section{Model and problem formulation}	\label{sec:Model}

We capture the (inter-)dependence of security among the agents 
using an undirected graph, which we call the {\em dependence
graph}. A node or vertex in the graph corresponds to an agent (e.g., 
an individual or organization), and an undirected edge between nodes 
$n_1$ and $n_2$ indicates 
interdependence of their security. We interpret an undirected edge
as a pair of directed edges pointing in the opposite directions
with an understanding that a directed edge from node $n_1$ to node
$n_2$ indicates that the security of node $n_1$ affects that of
node $n_2$ in the manner we explain shortly. When there is
an edge between two nodes, we say that they are {\em immediate}
or {\em one-hop} neighbors or, simply, neighbors when it is clear.

We model the interaction among the agents as a {\em noncooperative 
game}, in which players are the agents in the dependence 
graph.\footnote{We will use the words {\em agents}, {\em nodes} and 
{\em players} interchangeably hereafter.} 
This is reasonable because, in many cases, it may be difficult for  
agents to cooperate with each other and take coordinated security measures 
against attacks. In addition, even if they could coordinate their 
actions, they would be unlikely to do so in the absence of clear 
incentives for coordination. 

We are interested in scenarios where the number of agents is large. 
As mentioned earlier, 
modeling detailed {\em microscale} interactions among many 
agents in a large network and analyzing ensuing games is difficult;
the number of possible strategy profiles typically increases
exponentially with the number of players and finding the NEs of 
noncooperative games is often challenging even with a moderate number 
of players (the curse of dimensionality).

For analytical tractability, we employ a {\em population game} model
with a continuous action space. 
Population games provide a unified framework and tools for studying
{\em strategic interactions} among {\em a large number of agents} 
under following assumptions \cite{Sandholm}. 
First, the choice of an individual agent 
has very little effect on the payoffs of other agents. 
Second, the payoff of each agent depends only on the {\em 
distribution} of actions chosen by the members of each population.
For a detailed discussion of population games, 
we refer an interested reader to the manuscript by 
Sandholm~\cite{Sandholm}. We will follow the 
language of \cite{Sandholm} throughout the paper. 

Our population game does not explicitly capture the {\em 
link level} interactions between every pair of neighbors in
a fixed dependence graph. Instead, it is a simplification of complicated
reality and only attempts to capture the {\em average}
or {\em mean} behavior of agents with varying degrees. 
A key advantage of this model is that 
it provides a {\em scalable} model that enables us to study the 
{\em aggregate} behavior of the
players, resulting NEs and social optima, and their relation, 
{\em regardless of} the number of agents.

\subsection{Population game} 	\label{subsec:PG}

We assume that the maximum degree among all players in the 
dependence graph is $D_{\max} 
< \infty$. For each $d \in \{1, 2, \ldots, D_{\max}\} =: \cD$, 
$s_d$ denotes the {\em size} or {\em mass} of population 
consisting of players with degree 
$d$, and the population size vector 
${\bf s} := \big( s_d; \ d \in \cD \big)$ tells us the 
sizes of populations with varying degrees. 

The population size $s_d$ does {\em not} necessarily represent the 
{\em number} of agents in population $d$; instead, an implicit 
modeling assumption of a population game is
that each population consists of so many agents that a population 
$d \in \cD$ can be approximated as a {\em continuum} of 
{\em mass} or {\em size} $s_d \in (0, \infty)$. Hence, the
{\em ratios} of population sizes are more important than the
assumed population sizes, as it will be clear. 
\myskip

\noindent
{\bf i. Pure action/strategy space --}
All players have the identical (pure) action space ${\cal A} := 
[I_{\min}, I_{\max}] \subset \R_+ := [0, \infty)$, 
where $0 \leq I_{\min} < I_{\max} < \infty$. 
Each pure action in $\cA$ represents the amount that
an agent invests in security to protect itself. We denote the
set of distributions over $\cA$ by ${\cal P}_\cA$. 
\myskip

\noindent 
{\bf ii. Population states and social state --}
The {\em population state} of population $d$ is given by ${\bf x}_d 
\in {\cal P}_\cA$. In other words, given any Borel subset
${\cal S} \subseteq \cA$, $\bx_d({\cal S})$ tells us the {\em fraction}
of population $d$ whose security investment lies in ${\cal S}$. 
The {\em social state} consists of the population states of all 
populations and is denoted by $\bx = (\bx_d; d \in \cD) 
\in {\cal P}_\cA^{D_{\max}} =: \cX$.  
\myskip

\noindent 
{\bf iii. Cost function --} The cost function of the game is 
determined with the help of a function $C: \cX \times \cD \times 
\cA \times \R_+^{D_{\max}} \to \R$. The 
interpretation is that, when the population size vector
is $\bs$ and the social state is $\bx$, the cost of 
a player with degree $d$ investing $a$ in security is 
equal to $C(\bx, d, a, \bs)$. As we will show shortly, in addition 
to the cost of security investments, our cost function 
also reflects the (expected) losses from attacks. 

As mentioned earlier, we are interested in exploring a possible 
means of improving overall security via internalization of
externalities. Obviously, the externalities produced by players
will depend on the (properties of) dependence graph because
it determines how the security of one agent influences that of
other agents.  In order to capture this, we model 
two different types of attacks players suffer from -- {\em direct} 
and {\em indirect} attacks. 
While the underlying dependence graph does not affect the first 
type of attacks, it influences the latter type, 
thereby allowing us to capture the desired {\em network 
effects} shaped by it. 

{\em a) Direct attacks: }
We assume that attacker(s) launch an attack against each player 
with probability $\tau_A$, independently of other 
players.\footnote{Our model can be altered to capture the 
intensity or frequencies of attacks instead, with appropriate
changes to cost functions of the players.} We call
this a {\em direct} attack. When a player experiences a direct
attack, the realized cost depends on its security investment; 
when a player adopts action $a \in \cA$, it is infected with 
probability $p(a) \in [0, 1]$. Also, each time a player is
infected, it incurs on the average a loss of $L$. 
Hence, the expected
loss due to a single attack for a player with security investment 
of $a$ is $L(a) := L \cdot p(a)$.

{\em b) Indirect attacks: }
Besides the direct attack by a malicious attacker, a player also 
experiences {\em indirect} attacks from its neighbors that have
sustained a successful attack and are infected.  For instance, 
malware that successfully infects a user may scan
the user's address book and either send a malicious email to those 
on the contact list or forward the list to a server that sends out 
malevolent emails to those on the list. 

We assume that an infected agent launches an indirect attack on
each of its immediate neighbors (along the directed edges to
the neighbors) with probability $\beta_{IA} \in (0, 1]$
independently of each other. We call $\beta_{IA}$ indirect attack 
probability (IAP). When a player suffers an indirect attack, 
it is infected with the same probability $p(a)$, where $a$ is
its security investment. Moreover, 
a player infected by an indirect attack can also 
transmit the infection to its neighbors when an infection can 
propagate more than one hop. 

Baryshnikov \cite{Bary} showed that, under some technical 
assumptions, the infection probability (which the author
called {\em security breach probability})
is a {\em log-convex} (hence, strictly convex) 
decreasing function of 
the investments. The basic intuition behind this finding is 
the following: Suppose that there are many independent security 
measures an agent can employ to protect itself
(e.g., installation of security software, traffic monitoring).  
When the agent is free to choose any collection of security 
measures subject to its budget, it should choose the most
effective combination of security measures which minimizes 
its security
breach probability, leading to a diminishing return of 
increasing security investments~\cite{Gordon2015}.   

Based on this finding, we introduce the following
assumption on the infection probability $p(a)$, $a \in \cA$.
A similar assumption was used in other studies
(e.g., \cite{Gordon2015, Zhao2009}). 
\myskip

\begin{assm} 	\label{assm:pa}
The infection probability $p: \cA \to [0, 1]$ 
is {\em continuous, strictly convex and decreasing}. Moreover, it is 
continuously differentiable over $(I_{\min}, I_{\max})$. 
\myskip
\end{assm}

The IAP affects the {\em local} spreading behavior of infections. 
Unfortunately, the dynamics of infection propagation in a network 
depend on the details of
underlying dependence graph, which are difficult to obtain
or model faithfully. In order to skirt this difficulty, instead of
attempting to model the detailed, microscale 
dynamics of infection transmissions
between players, we abstract out the {\em (security) risks} seen by 
players using the (average) number of attacks a 
player sees from a single neighbor as explained below.  

$\bullet$ {\bf{\em Node degree distribution and weighted degree distribution --}}
We denote the mapping that yields the degree distribution of 
populations by ${\bf f}: \R_+^{D_{\max}} \to \Delta_{D_{\max}}$, 
where $\Delta_{D_{\max}}$ denotes the probability simplex 
in $\R^{D_{\max}}$ and 
\beqa
f_d({\bf s}) = \frac{ s_d }{ \sum_{d' \in \cD} s_{d'} }, 
	\ {\bf s} \in \R_+^{D_{\max}} \mbox{ and } d \in \cD, 
		\label{eq:bf}
\eeqa
is the fraction of total population with degree $d$. 
Similarly, define ${\bf w}: \R_+^{D_{\max}} \to \Delta_{D_{\max}}$, 
where 
\beqa
w_d({\bf s}) 
\myeq \frac{d \cdot s_d}{\sum_{d' \in \cD} d' \cdot s_{d'}} \lb
\myeq \frac{ d \cdot f_d(\bs) }{ d_{\avg}(\bs) }, 
	\ {\bf s} \in \R_+^{D_{\max}} \mbox{ and } d \in \cD,  
	\label{eq:bw}
\eeqa
and $d_{\avg}(\bs) := \sum_{d' \in \cD} d' \cdot f_{d'}({\bf s})$
is the average degree of nodes. 
From the above definition, ${\bf w}$ gives us the {\em weighted} 
degree distribution of populations, where
the weights are the degrees. 

It is easy to show
that both ${\bf f}$ and ${\bf w}$ are scale invariant. 
In other words, ${\bf f}({\bf s}) = {\bf f}(\phi \cdot {\bf s})$ 
and ${\bf w}({\bf s}) = {\bf w}( \phi \cdot {\bf s})$ for all 
$\phi > 0$. 
When there is no confusion, we write $\bff$, ${\bf w}$, and 
$d_{\avg}$ in place of $\bff({\bf s})$, ${\bf w}({\bf s})$, 
and $d_{\avg}(\bs)$, respectively. 

We first clarify the role of mapping ${\bf w}$. 
Suppose that we fix a social state ${\bf x} \in \cX$ and  
choose a player at random. 
In addition, assume that the dependence graph
is {\em neutral} \cite{Newman2002, Newman2003}, 
i.e., there are no correlations between the
degrees of two neighbors.  In this case, 
the probability that a randomly picked neighbor 
of the player has degree $d \in \cD$ can be approximated
using $w_d$ because it is proportional to the degree $d$ 
\cite{Callaway}.\footnote{When a network is either assortative
or disassortative, this assumption does not hold. 
The effects of assortativity on security is studied in 
\cite{La_assortativity}.} 
Hence, we can approximate the probability that the neighbor has 
degree $d$ and its security investment belongs to $B \subseteq \cA$ 
using $w_d \cdot x_d(B)$. 

This degree-based model is known as the Chung-Lu
model in the literature~\cite{ChungLu} and has been 
used extensively in other existing studies, e.g., \cite{Gleeson2007, 
Watts2002, Yagan2012}. 
In particular, Watts in his seminal paper~\cite{Watts2002} utilized
a similar degree-based model to study cascades of infection in a
network with a given degree distribution. He demonstrated that the
analytical results he derived using the generating function
method based on this model closely match the numerical results he 
obtained using random graphs.

$\bullet$ {\bf{\em Risk exposure --}}
Based on the above observation, we can model the average 
number of indirect attacks a node experiences from a single 
neighbor, assuming a {\em neutral} dependence graph, 
as follows.
First, we approximate the (expected) total 
number of one-hop indirect attacks from the victims of successful 
{\em direct} attacks to their {\em immediate}
neighbors using
\beqan
\Gamma(\bx ; \bs)
& \myb := & \myb \tau_A \cdot \beta_{IA} \sum_{d \in \cD} 
	\big( d \cdot  s_d \cdot p_{d, \avg}(\bx) \big),  
	\label{eq:Gamma1}
\eeqan
where 
\beqan
p_{d, \avg}(\bx)
& \myb := & \myb \int_{\cA} p(a) \ \bx_d(da)
\eeqan
is the probability that a randomly selected node of degree $d$ 
will be infected when it experiences an attack. Hence, $\tau_A \cdot 
\beta_{IA} \cdot d \cdot p_{d, \avg}(\bx)$ can be taken to be the average 
number of one-hop indirect attacks that nodes of degree $d$
inflict on their immediate neighbors.

Define $\gamma_{\avg}(\bx; \bs)$ to be 
the average number of {\em one-hop} indirect attacks {\em per 
directed edge} or, equivalently, 
the {\em fraction} of directed edges employed
to transmit {\em one-hop} indirect attacks. 
Then, $\gamma_{\avg}(\bx; \bs)$   
is equal to $\Gamma(\bx; \bs)$ divided by the total number
of directed edges in the dependence graph, i.e., 
\beqa
&& \myhb \gamma_{\avg}(\bx; \bs) 
	= \frac{\Gamma(\bx; \bs)}{\sum_{d' \in \cD} d' \cdot s_{d'}}  \lb
\myeq \frac{ \tau_A \cdot \beta_{IA} \sum_{d \in \cD} 
	\big( d \cdot f_d(\bs) \cdot p_{d, \avg}(\bx) \big)}
	{d_{\avg}(\bs)}
		\label{eq:gamma_avg1} \\
\myeq \tau_A \cdot \beta_{IA}
	\sum_{d \in \cD} \big( w_d(\bs) \cdot p_{d, \avg}(\bx) \big),  
	\nonumber
\eeqa
where the second and third equalities follow from (\ref{eq:bf}) 
and (\ref{eq:bw}). Since each directed edge corresponds
to a neighbor of some node, $\gamma_{\avg}(\bx; \bs)$ is also 
the {\em likelihood that a node will see a one-hop indirect attack 
from a neighbor on a randomly selected directed edge}.

When infections can propagate more than one-hop, i.e., beyond
immediate neighbors, the manner in which they can spread through
the network will depend on the {\em detailed structure} of the 
dependence graph. As mentioned earlier, it is difficult to model the 
dynamics of propagation accurately and makes 
a mathematical analysis challenging, if
possible at all. For this reason, we do not attempt to model
the detailed dynamics of infection transmissions. 
However, it is reasonable to expect that, even in general
settings, the average number of {\em indirect} attacks a node 
sees from a single neighbor (including both one-hop and multi-hop
indirect attacks) is increasing in $\gamma_{\avg}(\bx; \bs)$. 
This is formally stated by the following assumption. 
\myskip

\begin{assm}	\label{assm:exposure1}
The average number of indirect attacks that a node experiences
from a {\em single} neighbor, which we denote by $e(\bx; \bs)$,
is equal to $g(\gamma_{\avg}(\bx; \bs))$ for some function 
$g: \R_+ \to \R_+$ with $g(0) = 0$. The function $g$ is continuous, 
strictly increasing and differentiable over $(0, \infty)$. 
\myskip
\end{assm}

We call $e(\bx; \bs)$ the {\em risk exposure} at social state 
$\bx$ and employ it to quantify and compare the overall 
risk perceived by a node with a {\em fixed} degree, say 
$d \in \cD$, at different social states or under different 
network settings. In other words, we use it as a measure 
of {\em local} network security {\em from 
the viewpoint of a node with a fixed degree}.\footnote{To 
the best of our knowledge, one of challenges to studying 
network-level security is that there are no standard metrics 
experts agree on for quantifying network security.}
It is clear from (\ref{eq:gamma_avg1}) that $e(\bx; 
\bs)$ reflects the node degree distribution ${\bf f}(\bs)$.  
\myskip

{\bf Example:} Power function of $\gamma_{\avg}(\bx; \bs)$
-- In this case, the risk exposure is equal to 
\beqa
\myb e({\bf x}; \bs) 
\myeq \kappa \cdot \gamma_{\avg}(\bx; \bs)^b \lb
\myeq \kappa^+  \Big( \sum_{d \in \cD}  w_d(\bs) \cdot 
		p_{d, \avg}(\bx) \Big)^b 
	\label{eq:Exposure}
\eeqa
for some $\kappa, b > 0$, 
where $\kappa^+ :=  \kappa (\tau_A \cdot \beta_{IA})^b$.
Note that $\sum_{d \in \cD} w_d(\bs) \cdot p_{d, \avg}(\bs)$
is the probability that a randomly chosen neighbor of a node
(with degree distribution $\bw(\bs)$ as explained before)
will be infected by a single attack. Hence, it measures how
{\em vulnerable} a neighboring node is to an attack on 
the average.

The constant $\kappa$ can be used to capture how far
an infection spreads (e.g., the number of hops) 
on the average. In other words, the larger $\kappa$
is, the farther an infection disseminates through 
a network. On the other hand, the parameter $b$ determines
how sensitive the risk exposure is to the likelihood of 
a neighbor falling victim to an attack, i.e., the vulnerability
of a neighbor mentioned in the previous paragraph.

We assume that the total cost of a player due to multiple successful 
attacks it suffers is additive. 
The additivity of costs from different attacks is reasonable in 
many scenarios, including the earlier examples of the spread
of malware 
and corporate networks; each time a user is infected or its ID 
is stolen, the user will need to spend time and incur expenses 
to deal with the problem. Similarly, when a corporate network 
is breached, besides any financial losses or legal expenses, the 
network operator will need to assess the damages and take 
corrective measures.

Based on this assumption, we adopt the following cost function for our 
population game: for any given social state ${\bf x} \in \cX$, the 
cost of a node with degree $d \in \cD$ investing $a \in \cA$ in 
security is equal to 
\beqa
\myhb C({\bf x}, d, a, \bs) 
\myeq \left( \tau_A + d \cdot e({\bf x}; \bs) \right) L(a) 
	+ a. 
	\label{eq:Cost}
\eeqa
Note that $\tau_A + d \cdot e(\bx; \bs)$ is the expected number 
of attacks that a node of degree $d$ 
experiences, including both direct and indirect attacks. 

We focus on Nash equilibria (NEs) of population games as an 
approximation to nodes' behavior in practice. 
For every $d \in \cD$, define a mapping $A^{\opt}_d: \cX
\times \R_+^{D_{\max}} \to \mathcal{B}(\cA)$, where 
$\mathcal{B}(\cA)$ is the set of Borel subsets of $\cA$ and 
\beqan
A^{\opt}_d(\bx, \bs) 
:= \left\{ a \in \cA \ | \ C(\bx, d , a, \bs) 
	= \inf_{a' \in \cA} C(\bx, d, a', \bs) \right\}. 
\eeqan
For a fixed population size $\bs$, 
a social state ${\bf x}^\star$ is an NE if  $\bx^\star_d(
A^{\opt}_d(\bx^\star, \bs)) = 1$ for all $d \in \cD$. 
In other words, (almost) every player adopts a best response. 

Key questions we are interested in exploring with help of the 
population game model include:
{\em (i) Is there any structural relation between an NE of a 
population game and a social optimum? (ii) If so, what is the
relation and how does it depend on system parameters? 
(iii) In addition, how can we take advantage 
of the relation in order to improve network security?}
We shall offer some answers to these important questions in 
the following sections.

\section{Preliminaries}
	\label{sec:Preliminary}
	
Before we proceed, let us point out a useful observation. From 
(\ref{eq:gamma_avg1}) - (\ref{eq:Cost}), it is clear that 
the cost function is identical for two population size vectors 
$\bs^1$ and $\bs^2$ with the same node degree distribution, 
i.e., ${\bf f}(\bs^1) = {\bf f}(\bs^2)$.  This scale invariance 
property of the cost function implies that the set of NEs is 
identical for both population size vectors. It is also consistent
with an earlier comment that the ratios of population sizes
are more relevant than the absolute population sizes. As a result, 
it suffices to consider population size vectors 
whose sum is equal to one, 
i.e., $\sum_{d \in \cD} s_d = 1$. For this reason, without loss
of generality we impose the
following assumption in the remainder of the paper. 
\\ \vspace{-0.1in}

\begin{assm} 	\label{assm:0}
The population size vectors are normalized so that
the total population size is equal to one. 
\\ \vspace{-0.1in}
\end{assm}

Keep in mind that, under Assumption~\ref{assm:0}, the 
node degree distribution ${\bf f}(\bs)$ is equal to the
population size vector ${\bf s}$, i.e., ${\bf f}(\bs) = \bs$.

For each $r \in \R_+$, let $I^{\opt}(r)$
be the set of optimal investments for a player when it sees 
$r$ expected attacks. In other words, 
\beqan
I^{\opt}(r) 
\myeq \arg \min_{a \in \cA} \big( r \cdot L(a)  
	+ a \big). 
\eeqan
Under Assumption~\ref{assm:pa}, the 
optimal investment is unique, i.e., $I^{\opt}(r)$ is a singleton 
for all $r \in \R_+$. Hence, we can view $I^{\opt} : \R_+ \to \cA$ 
as a mapping that tells us the optimal security investment that a player 
should choose as a function of the total risk it perceives. 
This means that, at an NE $\bx^\star$, the population state 
$\bx^\star_d$ is concentrated on a single point, i.e., 
$\bx^\star_d( \{I^{\opt}( \tau_A + d  \cdot e(\bx^\star; \bs) )\} ) 
= 1$ for all $d \in \cD$. Moreover, one can show from 
the assumed strict convexity and continuous differentiability
of infection probability $p$ 
(Assumption~\ref{assm:pa}) that the optimal investment 
$I^{\opt}(r)$ is nondecreasing in $r$ and, if $I_{\min} 
< I^{\opt}(r') < I_{\max}$ for some $r' \in \R_+$, 
then $I^{\opt}(r) > I^{\opt}(r')$ 
(resp. $I^{\opt}(r) < I^{\opt}(r')$) for all $r > r'$
(resp. $r < r'$). 

Define the mapping $p^\star: \R_+ \to [0, 1]$ to be the 
composition of $I^{\opt}: \R_+ \to \cA$ and $p: \cA \to [0, 1]$, 
i.e., $p^\star(r) = p \circ I^{\opt}(r)$. Then, together
with the assumption that $p$ is strictly decreasing, 
the above observation means that $p^\star$ is nonincreasing. 
\myskip

\begin{coro}	\label{coro:mono}
If $0 \leq r_1 < r_2 < \infty$, we  have $p^\star(r_1) \geq p^\star(r_2)$. 
Furthermore, if $I_{\min} < I^{\opt}(r_1) < I_{\max}$, the 
inequality is strict, i.e., $p^\star(r_1) > p^\star(r_2)$.
\myskip
\end{coro}


\paragraph{Existence of a pure-strategy NE of a population 
game} 
Throughout the paper, we are often interested in 
cases in which the social state is degenerate, i.e., all 
players with the identical degree adopt the same action. We denote
the action chosen by population $d \in \cD$ by $a_d$, and 
refer to $\ba := (a_d; \ d \in \cD) \in \cA^{D_{\max}}$ as 
a {\em pure strategy profile}. 

With a little abuse of notation, we denote the risk 
exposure when a pure strategy profile $\ba$
is employed by 
\beqa
e(\ba; \bs)
& \myb := & \myb  g\Big( \tau_A \ \beta_{IA} 
	\sum_{d \in \cD} \big( w_d(\bs) 
	\cdot p(a_d) \big) \Big).
	\label{eq:ExposurePure}
\eeqa 
Recall that $\sum_{d \in \cD} w_d(\bs) \cdot p(a_d)
= \bw(\bs)^T {\bf p}(\ba)$, where ${\bf p}(\ba)
= (p(a_d); \ d \in \cD)$, is the vulnerability of a 
neighbor, i.e.,  the likelihood that a 
randomly chosen neighbor of a node will be infected
when attacked. We shall denote 
$\bw(\bs)^T {\bf p}(\ba)$ by $\rho(\ba; \bs)$
in the rest of paper, and rewrite $g(\tau_A \ \beta_{IA}
\ \rho(\ba; \bs))$ as $g^+(\rho(\ba; \bs))$. In other
words, we define the mapping $g^+: \R_+ \to \R_+$  
to be $g^+(z) = g(\tau_A \ \beta_{IA} \ z)$ for all 
$z \in \R_+$, and it tells us the risk exposure as a 
function of the neighbor vulnerability 
$\rho(\ba; \bs) = \bw(\bs)^T {\bf p}(\ba)$. 
\myskip

\begin{defn}	\label{defn:PSNE}
A pure strategy profile $\ba' \in \cA^{D_{\max}}$ is said
to be a pure-strategy NE if, for all 
$d \in \cD$,
\beqan 
a'_d = \arg \min_{a \in \cA} \big( (\tau_A + d \cdot 
	e(\ba'; \bs)) L(a) + a \big). 
\eeqan
\end{defn}

\begin{lemma}	\label{lemma:existence}
For every population size vector $\bs \in \R_+^{D_{\max}}$, 
there exists a pure-strategy NE of the corresponding
population game. 
\end{lemma}
\begin{proof}
A proof of the lemma is provided in Appendix
\ref{appen:existence}. 
\end{proof}

As discussed earlier, under Assumption~\ref{assm:pa}, for any 
NE of a population game, say $\bx^{NE}$, there exists 
a pure strategy profile $\ba^+$ such that $\bx^{NE}(\{\ba^+\})
= 1$; once the risk exposure 
$e(\bx^{NE}; \bs)$ is fixed at the NE, every population has
a unique optimal investment that minimizes its cost given by
(\ref{eq:Cost}) as explained before. Obviously, $\ba^+$ is
a pure-strategy NE by definition. This tells us that all NEs 
of population games are pure-strategy NEs.

\section{Main Results: Internalization of Externalities}	
\label{sec:Internalization}

As mentioned in Section~\ref{sec:Introduction}, 
an important question in IDS
is how one can encourage selfish agents to make adequate
investments in security in order to improve overall
security. In other words, what types of incentive
mechanisms can be employed, for instance, with the help of
regulatory policies, to {\em incentivize} selfish agents so 
that they would invest more in security than they would otherwise
at an NE?
In this section, we offer a partial answer to this question
by exploring how we may be able to {\em internalize} 
externalities~\cite{Varian_Microeconomics}.

To this end, we first establish the uniqueness of the
(pure-strategy) NE. Then, we illustrate an intriguing
relation between the set of local minimizers of 
social cost and the set of pure-strategy NEs of a 
{\em related} population game. 
As we will demonstrate, 
the latter finding allows us to (i) establish the uniqueness 
of a social optimum under a mild technical condition 
{\em without requiring} the convexity of social 
cost and (ii) compare the risk 
exposures at the unique pure-strategy NE of the population 
game and a (local) minimizer of social cost. 
Our finding confirms that the 
selfish nature of players leads to under-investments in 
security (due to positive externalities from security
investments). 
Finally, based on the aforementioned relation between 
the local minimizers of social cost and the pure-strategy NEs 
of the related population game, we suggest a possible 
means of internalizing the externalities caused by 
players. 

Before presenting our results, we first point out an 
obvious consequence of Corollary~\ref{coro:mono}: 
the optimal investment by a player at an NE
is non-decreasing in its degree. This is because
the average number of attacks seen by a node is strictly
increasing in its degree according to the cost function 
in (\ref{eq:Cost}).

\subsection{Uniqueness of Nash equilibria of population games}	
\label{subsec:NE}

\begin{theorem} \label{thm:unique}
For a fixed population size vector $\bs \in \R_+^{D_{\max}}$, 
there is a unique (pure-strategy) NE of the population game. 
\myskip
\end{theorem}
\begin{proof}
A proof can be found in Appendix~\ref{appen:unique}. 
\end{proof}

Since there exists a unique NE of a population game 
and we measure the (local) security seen by the players
using the risk exposure, we can compare the security at 
the NEs under different settings or dependence graphs
and also to that of social optimum. Furthermore, we only
need to consider pure-strategy NEs for our 
study, which simplifies the analysis somewhat. 
For notational convenience, we denote the unique
pure-strategy NE (for a given population
size vector $\bs$) by $\ba^{NE}(\bs)$ or simply
by $\ba^{NE}$ when there is no confusion.

\subsection{Social optima}
	\label{subsec:SO}

We define the overall social cost at social state $\bx \in \cX$ 
to be the aggregate cost of all players, i.e., 
the sum of (i) expected losses from attacks and (ii) the 
costs of security investments by players, which is given by 
\beqa
C_T(\bx, \bs) 
\myeq \sum_{d \in \cD} s_d \left( \int_{\cA} C(\bx, d, a, \bs) 
	\ x_d(da) \right).
	\label{eq:GC_SocialCost} 
\eeqa
The goal of the social player (SP), e.g., policy designers, is 
to minimize the social cost in (\ref{eq:GC_SocialCost}). 

In general, the SP could choose a distribution over 
$\cA^{D_{\max}}$ to minimize the social cost. 
In our study, however, we restrict its 
action space to $\cA^{D_{\max}}$. In other words, the SP is allowed 
to choose only a single investment amount for each population. We 
make this assumption for the following two reasons: first, 
as shown in the previous subsection, for a fixed population size 
vector $\bs$, there exists a unique 
NE, namely pure-strategy NE $\ba^{NE}$. 
Secondly, perhaps more importantly, from the 
perspective of designing a policy, 
a sound policy should have similar requirements for
nodes that are alike (in our case, nodes with the same degree).  

Given a population size vector $\bs$ and a pure strategy profile 
$\ba \in \cA^{D_{\max}}$, we denote the social cost by
\beqa
\tilde{C}_T(\ba, \bs)
\myeq \sum_{d \in \cD} s_d \cdot \tilde{C}(\ba, d, \bs),
	\label{eq:SocialCost}
\eeqa
where $\tilde{C}: \cA^{D_{\max}} \times \cD \times \R_+^{D_{\max}}
\to \R$ with
\beqa
\tilde{C}(\ba, d, \bs)
\myeq (\tau_A + d \cdot e(\ba; \bs)) L(a_d) + a_d. 
	\label{eq:CostPure}
\eeqa
Then, the SP is interested in solving the following constrained 
optimization problem.
\\ \vspace{-0.1in}

{\bf SP-OPT PROBLEM:}
\beqan
\min_{\ba \in \cA^{D_{\max}}} \ \tilde{C}_T(\ba, \bs)
\eeqan
Generally, without imposing additional conditions, 
the SP-OPT PROBLEM need not be convex and its
minimizer is not guaranteed to be unique; this is because
$e(\ba; \bs) = g^+(\rho(\ba; \bs))$ 
is not assumed to be convex in $\rho(\ba; \bs)
= {\bf w}(\bs)^T {\bf p}(\ba)$.

Let $\cA_{LM}(\bs)$ be the set 
of {\em local} minimizers of social cost. It turns out
that there is an interesting relation between $\cA_{LM}(\bs)$
and the set of pure-strategy NEs of a closely related population 
game with a modified cost function.

\subsection{Related population game with a modified cost function}
	\label{subsec:Relation}

Consider the following modified population game with a slightly
different cost function $\hat{C}: {\cal A}^{D_{\max}} \times \cD 
\times \R_+^{D_{\max}} \to \R$, where
\beqa
&& \myhb \hat{C}(\ba, d, \bs)
		\label{eq:CostMod} \\
\myeq \big( \tau_A + d ( e(\ba; \bs)
	+ \dot{g}^+(\rho(\ba; \bs)) \ \rho(\ba; \bs) ) \big) L(a_d) 
		+ a_d \lb
\myeq \big( \tau_A + d \ \vartheta(\rho(\ba; \bs)) \big) L(a_d) 
	+ a_d \nonumber
\eeqa
for all $d \in \cD$ and $\ba \in \cA^{D_{\max}}$, 
and $\vartheta: \R_+ \to \R_+$ with $\vartheta(z) = 
g^+(z) + \dot{g}^+(z) \cdot z$. 
The difference between the cost function in (\ref{eq:CostPure}) 
and the above modified cost function in (\ref{eq:CostMod}) 
is that the number of indirect attacks seen from a single 
neighbor is scaled up from $e(\ba; \bs)$ to $\vartheta(\rho(\ba; \bs))
= e(\ba; \bs) + \dot{g}^+(\rho(\ba; \bs)) \rho(\ba; \bs) 
\geq e(\ba; \bs)$. Clearly, 
providing the additional information regarding
$\dot{g}^+(\rho(\ba; \bs)) \rho(\ba; \bs)$ may be problematic
in practice. We will revisit this issue shortly. 

Even when $g^+$ is strictly increasing, $\vartheta$ 
might not be strictly increasing. For example, 
suppose $g^+(z) = 1 - e^{-z}$, $z \in \R_+$. It is clear
that $g^+$ is strictly increasing. However, 
$\dot{\vartheta}(z) = 2 \dot{g}^+(z) + z \cdot \ddot{g}^+(z) 
= (2 - z) e^{-z}$. Thus, $\dot{\vartheta}(z) < 0$
when $z > 2$, and $\vartheta(z)$ is not monotonically 
increasing in $z$. 
As a result, unlike in the original population game, 
a (pure-strategy) NE of the modified population game is not 
guaranteed to be unique. We denote by $\cA_{NE}^{mc}(\bs)$ the 
set of pure-strategy NEs of this modified population game.
\\ \vspace{-0.1in}

\begin{theorem}	\label{thm:relation1}
Any local minimizer of social cost is a pure-strategy
NE of the modified population game. Therefore, 
\beqan
\cA_{LM}(\bs) \subseteq \cA_{NE}^{mc}(\bs).
\eeqan
\end{theorem}
\begin{proof}
A proof is provided in Appendix~\ref{appen:relation1}. 
\end{proof}

The theorem leads to a following corollary, which 
affirms that free riding by some players at NEs causes
a degradation in network security, which is measured by 
the risk exposure in this paper. 
\myskip

\begin{coro} 	\label{thm:GC_NI_comparison}
Fix the population size vector $\bs$.  Let $\ba^{NE} = 
\ba^{NE}(\bs)$ and $\ba^\star \in \cA_{LM}(\bs)$ be any
local minimizer of social cost. Then, 
$e( \ba^\star; \bs ) \leq e(\ba^{NE}; \bs)$. 
\myskip
\end{coro}

Corollary~\ref{thm:GC_NI_comparison} 
tells us that the risk exposure at the NE is greater than or 
equal to that of the {\em worst local minimizer}. 
The intuition behind the corollary is that because $\ba^\star$
is a local (or global) minimizer of social cost, 
Theorem~\ref{thm:relation1}
tells us that it is a pure-strategy NE of the modified 
population game in which all players perceive heightened risks,
thereby forcing them to invest more in security than they would
at the NE of the original population game.

\subsection{Internalization of externalities}
	\label{subsec:Internalization}

Theorem~\ref{thm:relation1} has another, perhaps more practically
important implication for realizing the internalization of externalities, 
under a following reasonable assumption. 
\myskip

\begin{assm}	\label{assm:vartheta}
The mapping $\vartheta: \R_+ \to \R_+$ is strictly
increasing. 
\myskip
\end{assm}

While Assumption~\ref{assm:vartheta} does not
always hold as we illustrated earlier, it is likely to hold
in practice; a {\em sufficient} condition for the assumption
is that $g^+$ is differentiable and 
$z \cdot g^+(z) =: h(z)$ is strictly convex (note that 
$\dot{h}(z) = z \cdot \dot{g}^+(z) + g^+(z) =
\vartheta(z)$). This sufficient condition 
holds, for instance, if $g^+$ can be written as a sum of 
power functions with positive exponents, i.e., 
$g^+(z) = \sum_{k} a_k \cdot z^{b_k}$ 
for some $a_k, b_k > 0$, or $g^+(z) = a \cdot 
\log(1+z)$ for some $a > 0$. 
Furthermore, it is reasonable to expect that 
$e(\ba; \bs) = g^+(\rho(\ba; \bs))$ is at least linearly
increasing in the average number of 
{\em one-hop} indirect attacks, 
namely $\gamma_{\avg}(\ba; \bs)$. Because $\rho(\ba; \bs) 
\propto \gamma_{\avg}(\ba; \bs)$), this suggests that
$g^+$ is likely to be at least linear in $\rho(\ba; \bs)$
and Assumption~\ref{assm:vartheta} is likely to be satisfied.
\myskip

\begin{coro}	\label{coro:unique}
Suppose that Assumption~\ref{assm:vartheta} holds. 
Then, $\cA^{mc}_{NE}(\bs)$ is a singleton, i.e., 
there exists a {\em unique} (pure-strategy) NE of the 
modified population game, $\ba^{mc}_{NE}(\bs)$. 
Consequently, there exists a unique global
minimizer of social cost, which coincides with
$\ba^{mc}_{NE}(\bs)$. 
\myskip
\end{coro}

The proof of the corollary is similar to that of Theorem
\ref{thm:unique} in Appendix~\ref{appen:unique}
and is omitted here. 

Theorem~\ref{thm:relation1}, together with Corollary
\ref{coro:unique}, sheds some light on the structural 
relation between the (global) minimizer of social cost and 
the pure-strategy NE of a related population game under
Assumption~\ref{assm:vartheta}.  
Also, it hints at how we might be able to strengthen network 
security, for example, by levying penalties or taxes on the 
players for indirect attacks they suffer. 

Finally, it tells us how the function $g^+$ (which determines 
the risk exposure $e(\ba; \bs)$ as a function of the
vulnerability of neighboring nodes)
influences how much penalty we should impose on agents for 
externalities they produce. For this reason, it 
reveals how the `shape' of the function $g^+$ affects 
the necessary penalty. 
Thus, this finding might be useful for designing a 
guideline for internalizing the externalities so as to 
improve overall security and reduce the social cost, 
assuming that we can estimate the function $g^+$. 
\myskip

{\bf Example: $g^+(z) = a \ z^b$, $a , b > 0$ --} 
Suppose that the function $g^+$ can be approximated by a 
power function over an interval of interest to us. 
In this case, Assumption
\ref{assm:vartheta} holds and Corollary~\ref{coro:unique}
tells us that the unique NE of the modified population game 
minimizes the social cost. Furthermore, $\vartheta(z) 
= a(1 + b) z^b = (1+b) g^+(z)$ and 
the risk seen by a node from a 
neighbor is scaled by $(1 + b)$ in the altered cost 
function in (\ref{eq:CostMod}). Therefore, the penalty
that we need to levy on agents in order to achieve the minimum
social cost at an NE can be easily obtained from the 
risk or losses that a node experiences due to indirect
attacks from its neighbors. To be precise, the 
penalty should be set to $b$ times the losses from indirect
attacks.\footnote{When most of attacks experienced by a node
are indirect attacks, the penalty would be approximately
equal to $b$ times the total losses. This might be a 
reasonable approximation 
for nodes with large degrees (so-called hubs) or when 
the network is highly connected.} 
 
This observation is quite intuitive; as $b$ increases,
the risk from/to a neighbor becomes more sensitive to 
a change in $\rho(\ba; \bs) = \bw(\bs)^T {\bf p}(\ba)$ 
or, equivalently, a
change in the security investments of players. 
As a result, when a player 
reduces its security investment, it causes higher 
negative externalities to its neighbors, thereby 
calling for a higher penalty.

\section{Effects of node degree distribution on 
	externalities and risk exposure}
	\label{sec:Property}
	
Recall that, for a given pure strategy profile $\ba$, 
the risk exposure, hence the necessary penalties
or taxes for
internalizing externalities, is a function of the node
degree distribution (via $\bw(\bs)^T 
{\bf p}(\ba)$). Hence, it is of interest to understand 
the effects of node degree distribution, which reflects
to some extent the level of interdependence in security 
among players, on their security investments at both the 
NEs of population games and social optima. 
\myskip

\begin{theorem} \label{thm:GC_NE_Monotonicity}
Let ${\bf s}^1$ and ${\bf s}^2$ be two population size vectors
that satisfy 
\beqa
\sum_{\ell = 1}^d w_\ell(\bs^1) 
\myleq \sum_{\ell = 1}^d w_\ell(\bs^2) 
	\ \mbox{ for all } d \in \cD
	\label{eq:thm11}
\eeqa
and $\ba^i = \ba^{NE}(\bs^i)$, $i = 1, 2$. 
Then, $e(\ba^1; \bs^1) \leq e(\ba^2; \bs^2)$.
Consequently, $a^1_d \leq a^2_d$ for all $d \in \cD$.
Furthermore, if the inequality in (\ref{eq:thm11}) is
strict for some $d \in \cD$ and $\ba^2$ lies in the 
interior of $\cA^{D_{\max}}$, then 
$e(\ba^1; \bs^1) < e(\ba^2; \bs^2)$.  
\end{theorem}
\begin{proof}
A proof is given in Appendix~\ref{appen:GC_NE_Monotonicity}. 
\end{proof}

If the inequality in (\ref{eq:thm11}) is strict for some 
$d \in \cD$, i.e., $\bw(\bs^1) \neq \bw(\bs^2)$, ${\bf w}(\bs^1)$
first-order stochastically dominates ${\bf w}(\bs^2)$
\cite{ShakedShan}. Thus, 
the first part of Theorem~\ref{thm:GC_NE_Monotonicity} states 
that, as the weighted node degree distribution becomes 
(stochastically) larger, the risk exposure at the NE declines. 
What is somewhat surprising 
is that the risk exposure falls even though every 
population $d$ invests {\em less} in security and is more
vulnerable to attacks because $a^1_d \leq a^2_d$ for
all $d \in \cD$. 

The intuition behind these 
perhaps counterintuitive findings is that, as the degrees 
of {\em neighbors} rise (with stochastically larger
weighted node degree distributions), a node with a fixed 
degree, say $d$, experiences diminished risk from
its neighbors because nodes with higher degrees, which see
larger risks and invest more in security, are better
protected as captured by Corollary~\ref{coro:mono}. 
As a consequence, the node enjoys greater positive externalities
from its neighbors and reduces its own security investments. 

Suppose that Assumption~\ref{assm:vartheta} holds and, 
hence, there exists a unique minimizer of social cost. 
We shall use $\ba^\star(\bs)$ to denote the unique
minimizer of the social cost for fixed population 
size vector $\bs$. When there is
no confusion, we simply denote it by $\ba^\star$.
The following theorem tells us that, analogously to the
finding in Theorem~\ref{thm:GC_NE_Monotonicity}, 
the risk exposure at the social optimum also
drops as the weighted node degree
distribution becomes larger. 
\\ \vspace{-0.1in}

\begin{theorem} \label{thm:GC_SO_Monotonicity}
Suppose that Assumption~\ref{assm:vartheta} is true. Let 
${\bf s}^1$ and ${\bf s}^2$ be two population size vectors
that satisfy 
\begin{equation*}
\sum_{\ell = 1}^d w_\ell(\bs^1) 
\leq \sum_{\ell = 1}^d w_\ell(\bs^2) 
	\ \mbox{ for all } d \in \cD. 
	\tag{\ref{eq:thm11}}
\end{equation*}
Then, $e\big( \ba^\star(\bs^1) ; \bs^1 \big) \leq 
e\big( \ba^\star(\bs^2) ; \bs^2 \big)$.
Moreover, if the inequality in (\ref{eq:thm11}) is
strict for some $d \in \cD$ and $\ba^\star(\bs^2)$ lies in the 
interior of $\cA^{D_{\max}}$, then the inequality is strict.  
\end{theorem}
\begin{proof}
Please see Appendix~\ref{appen:GC_SO_Monotonicity} for a proof. 
\end{proof}

Theorems~\ref{thm:GC_NE_Monotonicity} and 
\ref{thm:GC_SO_Monotonicity} state that the risk exposure
at both NEs and social optima tends to decline as the 
weighted node degree distribution becomes larger. Hence, our 
results indicate that both in the distributed case with selfish agents
and in the centralized case with SP, higher network connectivity 
will likely improve local network security measured by risk exposure
and, hence, the risk seen by a node with a fixed degree.

At the same time, our findings in \cite{La_TON_Cascade} reveal 
that the {\em global} network security may 
in fact deteriorate; a key observation of \cite{La_TON_Cascade}
is that as the weighted node degree distribution becomes larger, 
the cascade probability (i.e., the probability that a random single 
infection leads to a cascade of infection to a large number of nodes)
rises. Hence, together, these findings suggest that, as the weighted
node degree distribution gets larger, individual nodes with fixed 
degrees might perceive improved {\em local} 
security because the average number of 
attacks they see falls, whereas the {\em global} network security
degrades in that even a single successful infection may spread to 
a large number of other nodes in the network with higher probability. 

Corollary~\ref{coro:unique} offers some intuition as to
why the NEs of population games and social optima possess similar
properties shown in Theorems~\ref{thm:GC_NE_Monotonicity}
and \ref{thm:GC_SO_Monotonicity}; social optima can be viewed as
the pure-strategy NEs of related population games with modified cost 
functions.

The claims in Theorem~\ref{thm:GC_NE_Monotonicity} or
\ref{thm:GC_SO_Monotonicity} do not always hold when we replace 
the weighted node degree distributions in (\ref{eq:thm11}) 
with the node degree distributions of two population sizes. In other 
words, we can find two population size vectors $\tilde{\bs}^1$ and 
$\tilde{\bs}^2$ such that
\beqan
\sum_{\ell = 1}^d \tilde{s}^1_\ell 
\myleq \sum_{\ell = 1}^d \tilde{s}^2_\ell 
	\ \mbox{ for all } d \in \cD, 
\eeqan
but the claims in Theorem~\ref{thm:GC_NE_Monotonicity} 
or \ref{thm:GC_SO_Monotonicity} fail to hold. 

The following lemma provides a sufficient condition for
the condition (\ref{eq:thm11}). 
\\ \vspace{-0.1in}

\begin{lemma}	\label{lemma:1}
Suppose that two population size vectors $\bs^1$ and $\bs^2$ 
satisfy
\beqa
\frac{ s^2_d }{ s^1_d } 
	\geq \frac{ s^2_{d+1} }{ s^1_{d+1} } 
	\ \mbox{ for all } d \in \cD^- := \{1, 2, \ldots, D_{\max} - 1\}. 
	\label{eq:lemma0}
\eeqa
Then, the condition (\ref{eq:thm11}) in Theorems
\ref{thm:GC_NE_Monotonicity} and \ref{thm:GC_SO_Monotonicity} 
is satisfied.
\end{lemma}
\begin{proof}
Please see Appendix~\ref{appen:lemma1} for a proof. 
\end{proof}

The finding in Lemma~\ref{lemma:1} can be applied to several 
well known families of distributions. For example, consider a 
family of (truncated) power law degree distributions 
$\{ {\bf s}^{\alpha}; \ \alpha \in \R_+ \}$, where $s^\alpha_d 
\propto d^{-\alpha}$, $d \in \cD$. Suppose that $\alpha_1 \leq 
\alpha_2$. Then, one can easily show that ${\bf s}^{\alpha_1}$ 
and ${\bf s}^{\alpha_2}$ satisfy (\ref{eq:lemma0}) as 
follows:
\beqan
\frac{d^{-\alpha_2}}{d^{-\alpha_1}} 
\myeq d^{-(\alpha_2 - \alpha_1)}
\geq (d+1)^{-(\alpha_2 - \alpha_1)}
= \frac{(d+1)^{-\alpha_2}}{(d+1)^{-\alpha_1}} 
\eeqan
because $\alpha_2 - \alpha_1 \geq 0$.

\section{Numerical Results}
	\label{sec:Numerical}
	
In this section, we provide some numerical results to 
validate our main findings and to explore how the 
sensitivity of the risk exposure to security 
investments of agents (i.e., the derivative of 
$g^+$) and the infection probability function $p$
affect the inefficiency of NEs. 

For numerical studies, we use truncated power law degree 
distributions with varying parameter $\alpha$. The infection 
probability is equal to $p(a) = (1 + a)^{-\zeta}$
with $\zeta > 0$. The risk exposure function $g^+$ is given  
by $e(\ba; \bs) = g^+(\rho(\ba; \bs)) 
= 30 (\rho(\ba; \bs))^b$ for $b = 1.1$ 
and 2.0. Recall from the discussion in Section
\ref{subsec:Internalization} that the employed
risk exposure function $g^+$ satisfies Assumption
\ref{assm:vartheta} and the unique (pure-strategy) NE
of the modified population game coincides with the 
unique global minimizer of social cost.

The values of remaining parameters
are provided in Table~\ref{table:par}. We picked 
large $I_{\max}$ on purpose so that there is no 
budget constraint at the NE or social optimum.

\begin{table}[h]	
\begin{center}
\begin{tabular}{|c|c || c| c |}
\hline 
Parameter & Value & Parameter & Value \\
\hline
$D_{\max}$ & 20 & $L$ & 10  \\
$\tau_A$ & 0.7 & $\beta_{IA}$ & 1.0 \\
$\cA$ & [0, 1000] & $\zeta$ & 1.5 or 2.5 \\
$b$ & 1.1 or 2.0 & & \\
\hline
\end{tabular}
\end{center}
\caption{Parameter values.}
\label{table:par}
\end{table}

\begin{figure}[h]
\centerline{
\includegraphics[width=1.75in]{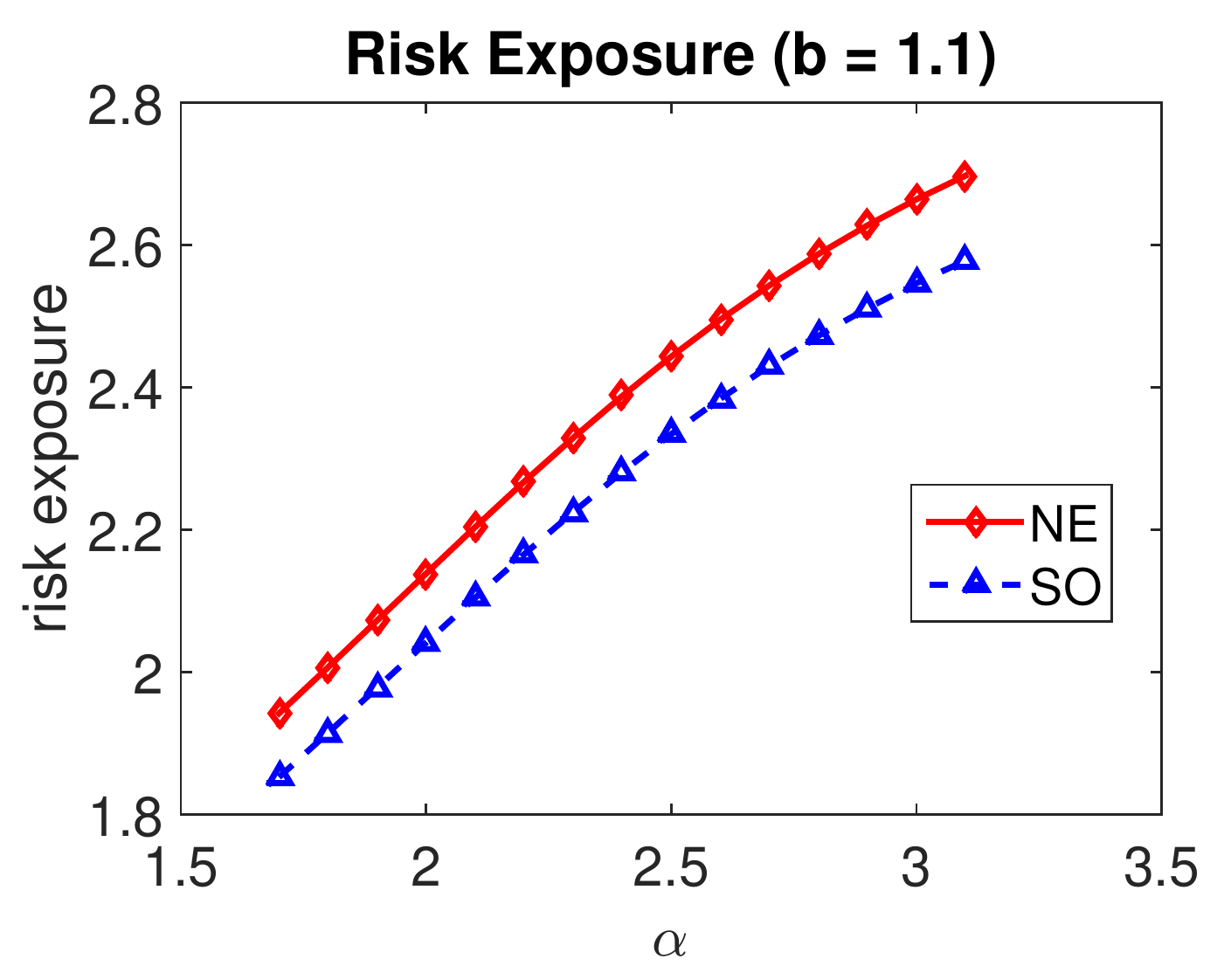}
\includegraphics[width=1.75in]{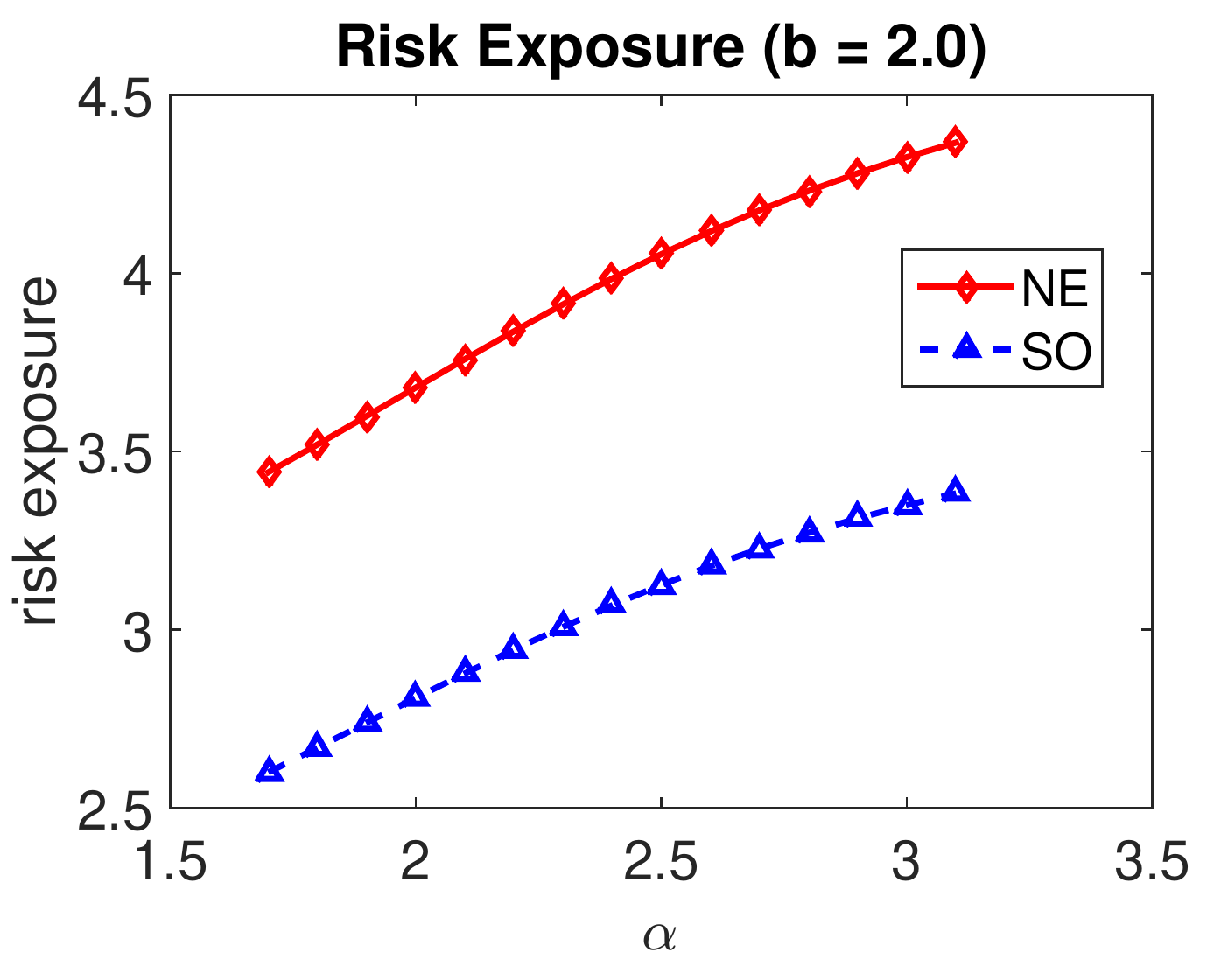}
}
\centerline{(a)}
\centerline{
\includegraphics[width=1.75in]{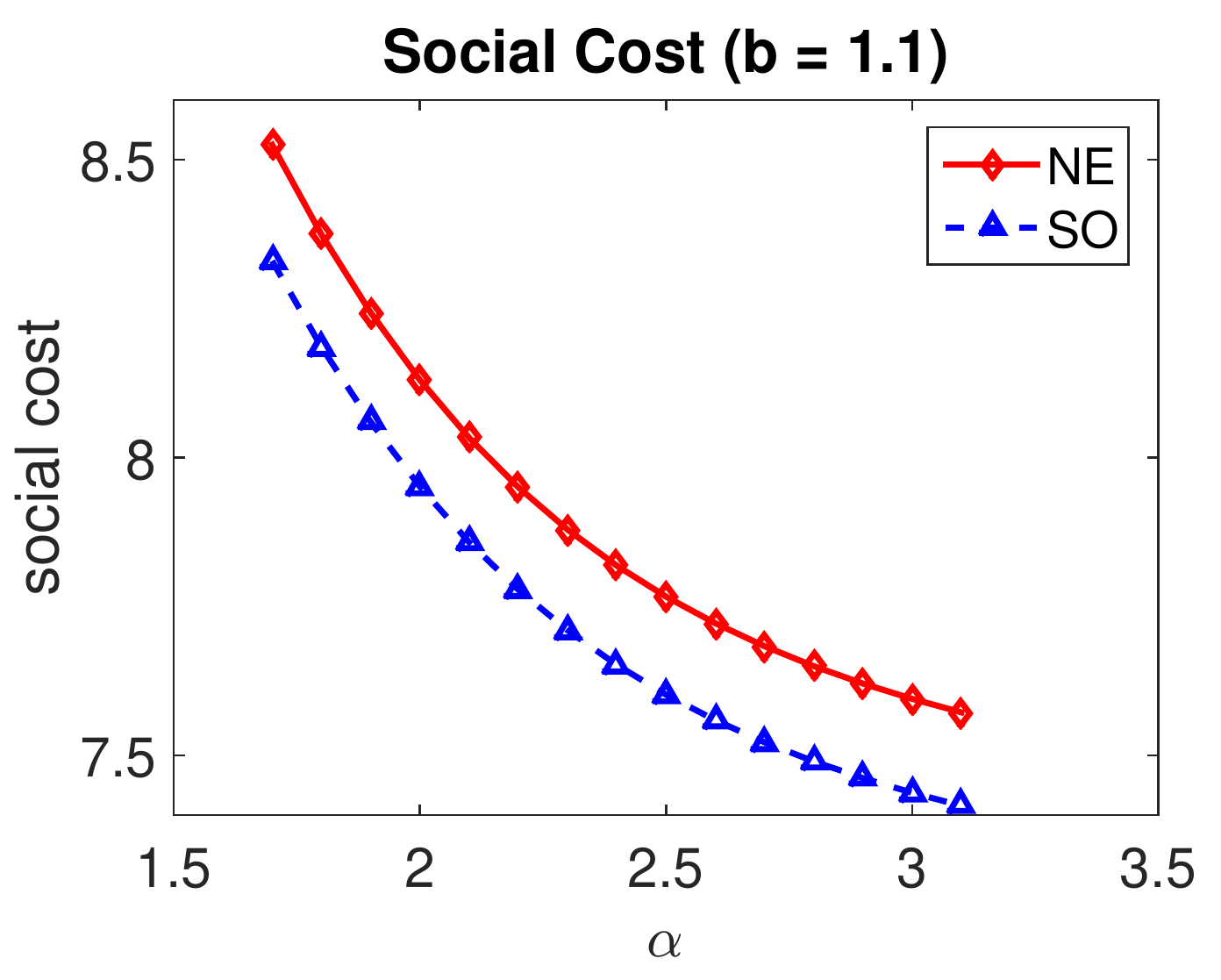}
\includegraphics[width=1.75in]{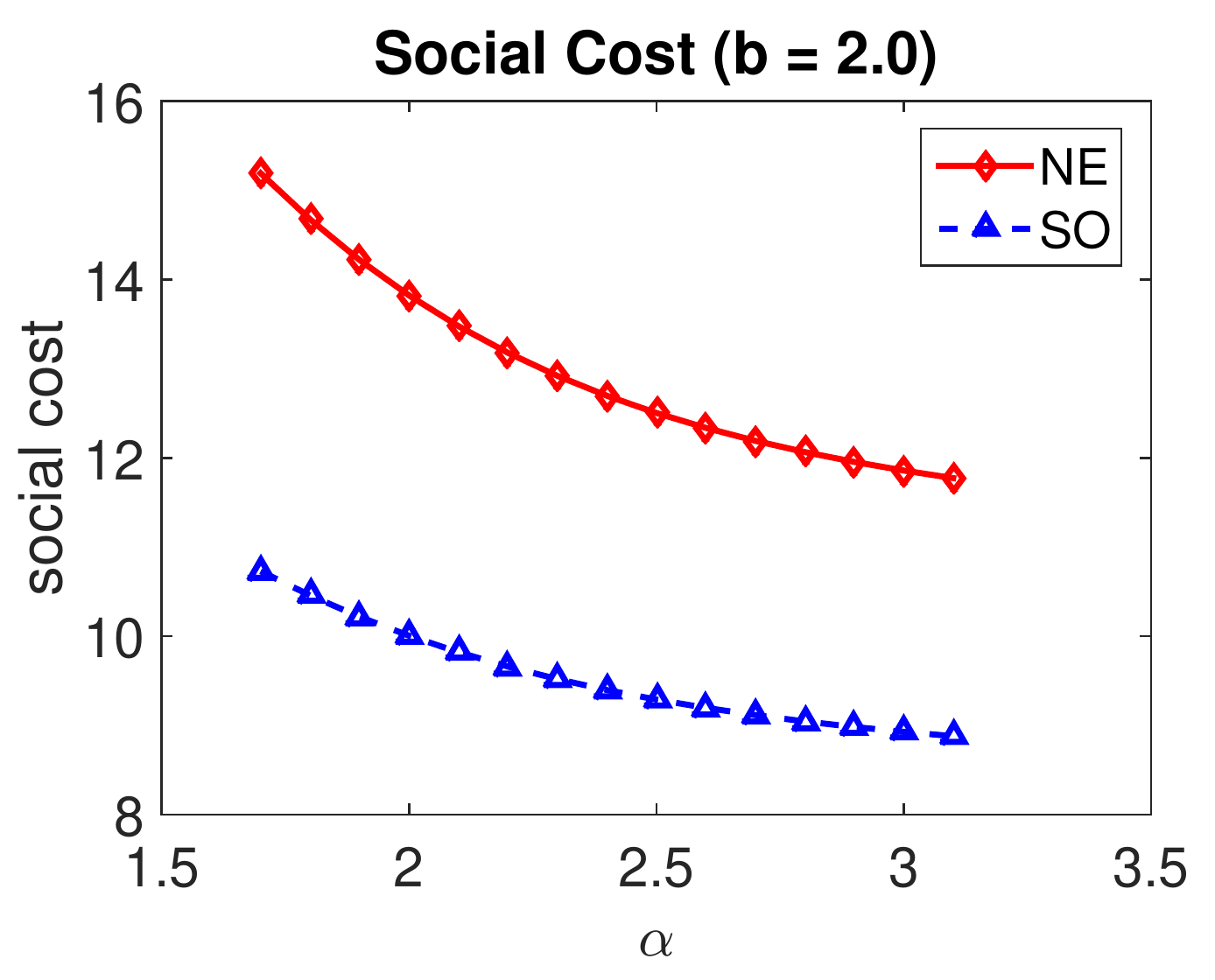}
}
\centerline{(b)}
\centerline{
\includegraphics[width=1.75in]{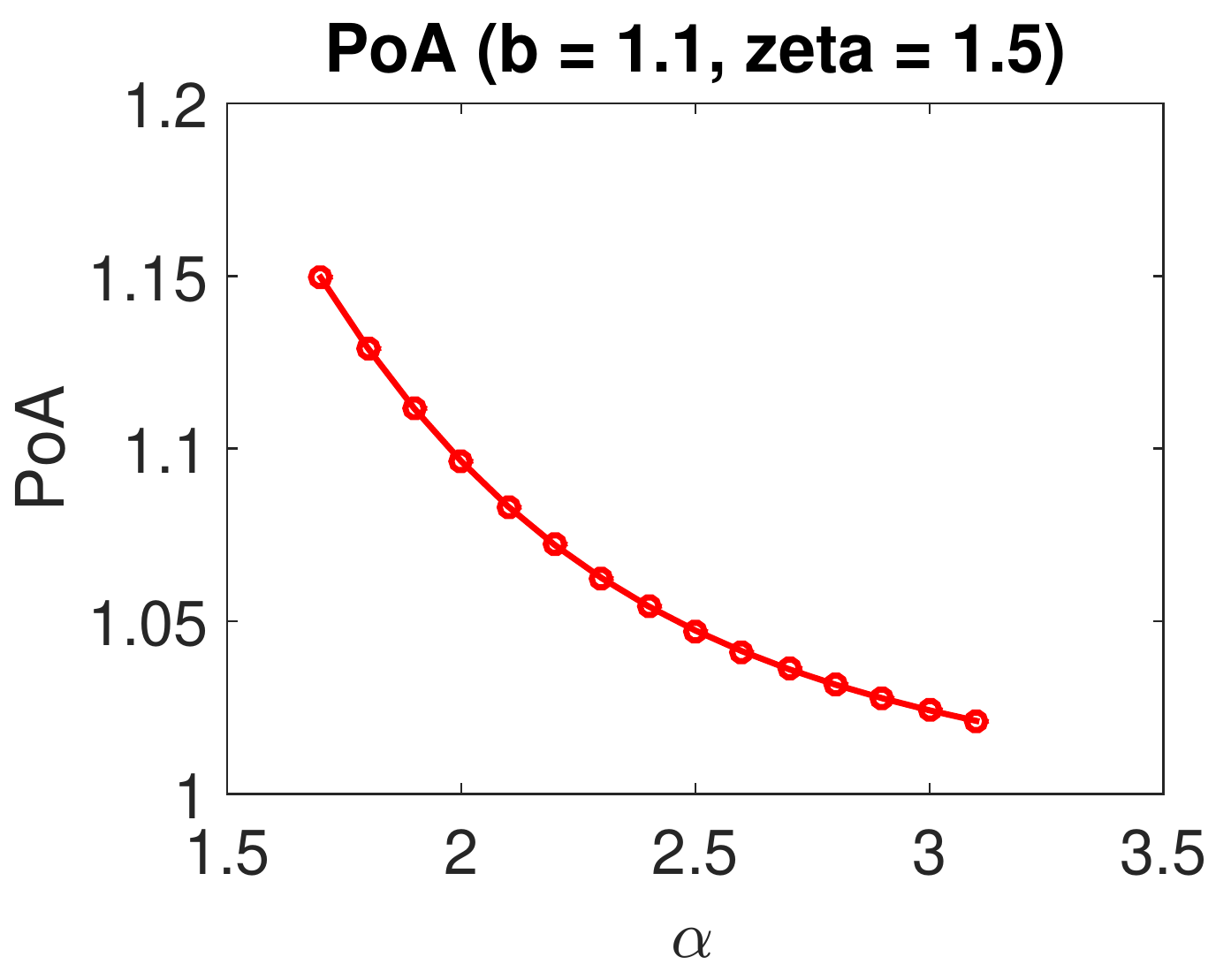}
\includegraphics[width=1.75in]{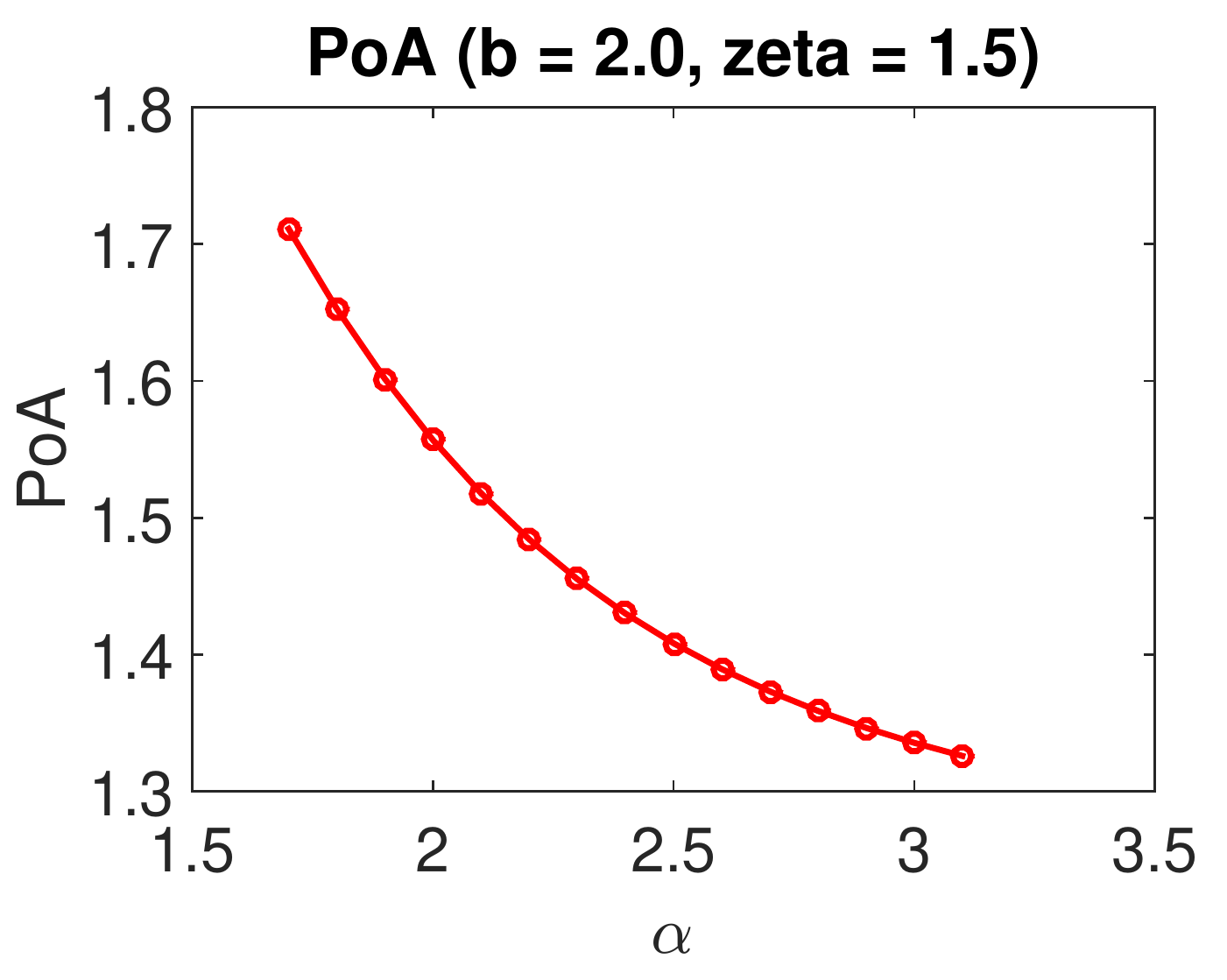}
}
\centerline{(c)}
\caption{(a) Risk exposure, (b) social cost and 
	(c) price of anarchy ($\zeta = 1.5$).}
\label{fig:15}
\end{figure}

Fig.~\ref{fig:15} plots the risk exposure and social cost 
at both NEs and social optima as well as the so-called
price of anarchy (PoA) for $\zeta = 1.5$ as we 
vary the value of power law parameter $\alpha$. The 
PoA is defined to be the ratio of the social cost 
at the worst NE to the minimum achievable social cost, 
and is a popular measure of the inefficiency of NE. 
Lemma~\ref{lemma:1} tells us that 
as $\alpha$ increases, the weighted degree
distribution $\bw$ becomes stochastically
smaller. Thus, Theorems~\ref{thm:GC_NE_Monotonicity}
and \ref{thm:GC_SO_Monotonicity} state that the 
risk exposure rises at both NEs and social optima. 
This is confirmed by Fig.~\ref{fig:15}(a). 

Our analytical findings, however, do not suggest how 
the social cost would behave with increasing $\alpha$. 
Fig.~\ref{fig:15}(b) shows that
the social cost in fact decreases with $\alpha$ in spite of
increasing risk exposure. This is because the 
underlying dependence graph becomes {\em less} 
connected and nodes have smaller degrees.

In addition, Figs.~\ref{fig:15}(b) and 
\ref{fig:15}(c) indicate that the 
gap between the NE and social optimum widens both in 
risk exposure and social cost as the risk
exposure (i.e., $g^+$) becomes more sensitive to the 
security investments when $b$ is raised to 2.0 
from 1.1, thereby causing higher PoA. 
This observation is intuitive; when the risk 
exposure is more sensitive to security investments, 
when agents make less security investments at the
NE compared to the social optimum, it leads to 
larger increase in risk exposure and social cost. 

Finally, Fig.~\ref{fig:15}(c) illustrates that 
the PoA is larger when the dependence graph is 
more densely connected. This observation is consistent
with that of 
\cite{La_TON_Cascade}: Theorem 7 of 
\cite{La_TON_Cascade} proves a tight upper
bound on PoA, which is an affine function of
the average node degree. Thus, as the dependence
graph becomes more connected, leading to a higher
average node degree, the NE becomes less efficient
in that the PoA escalates.

\begin{figure}[h]
\centerline{
\includegraphics[width=1.75in]{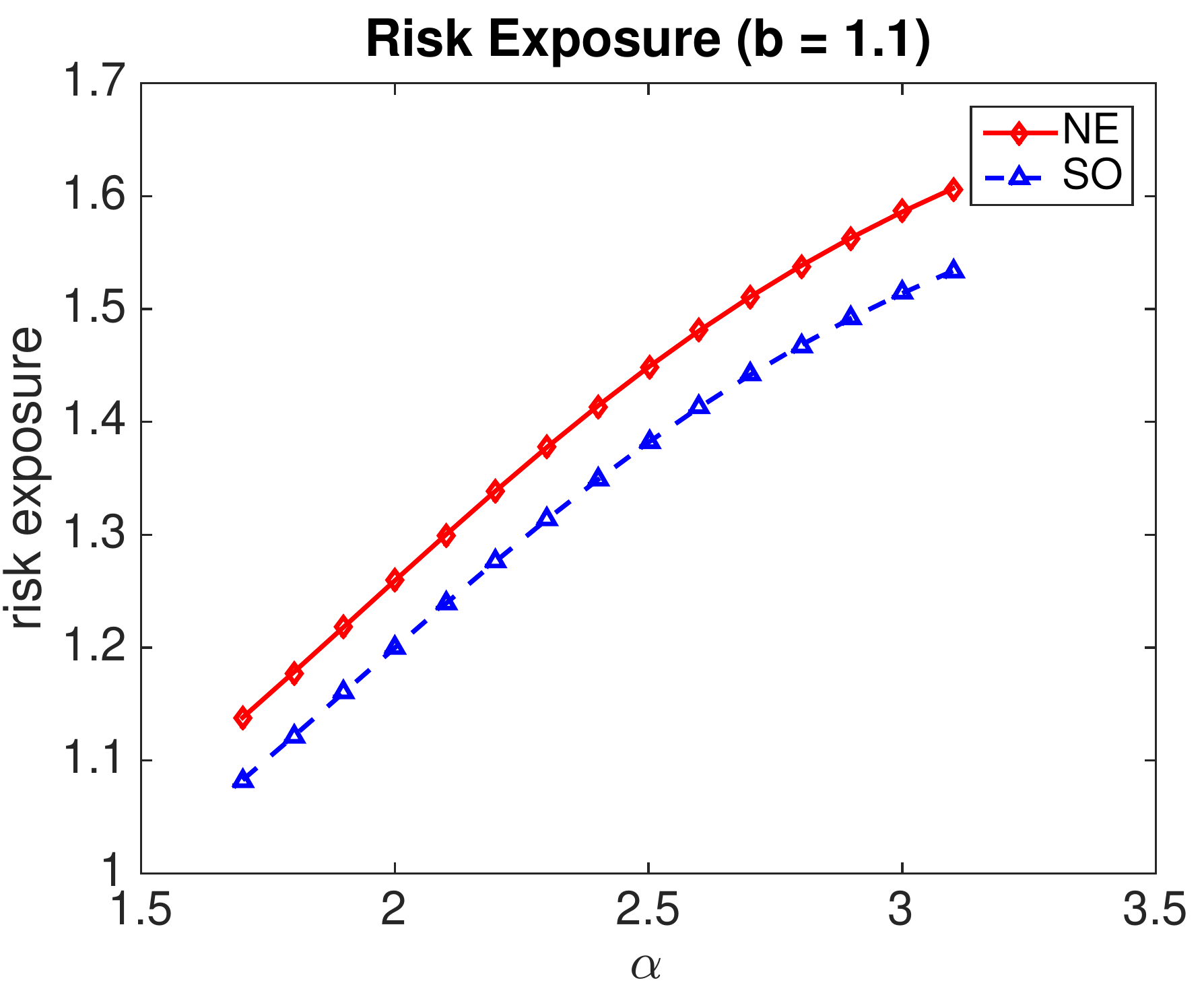}
\includegraphics[width=1.75in]{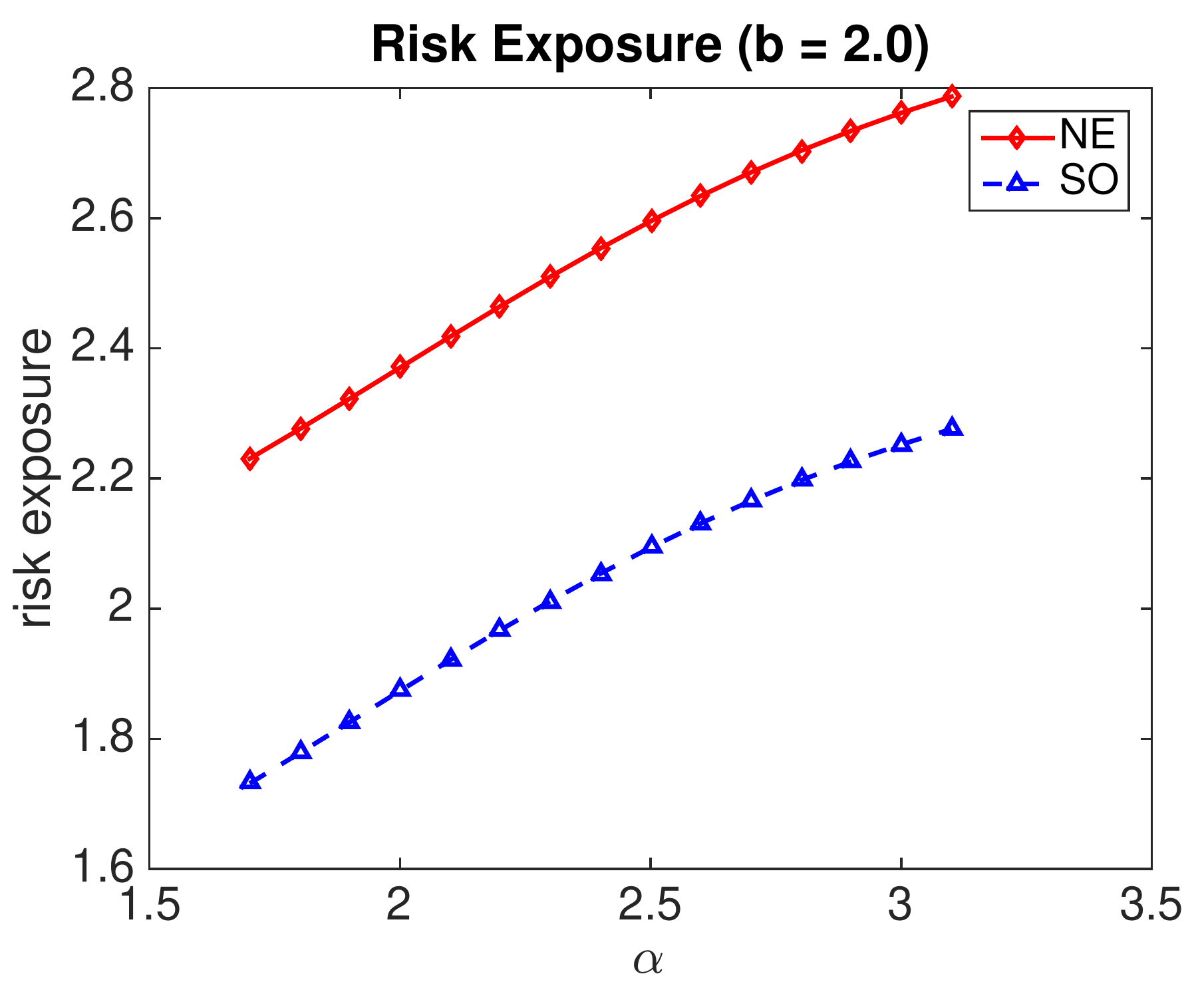}
}
\centerline{(a)}
\centerline{
\includegraphics[width=1.75in]{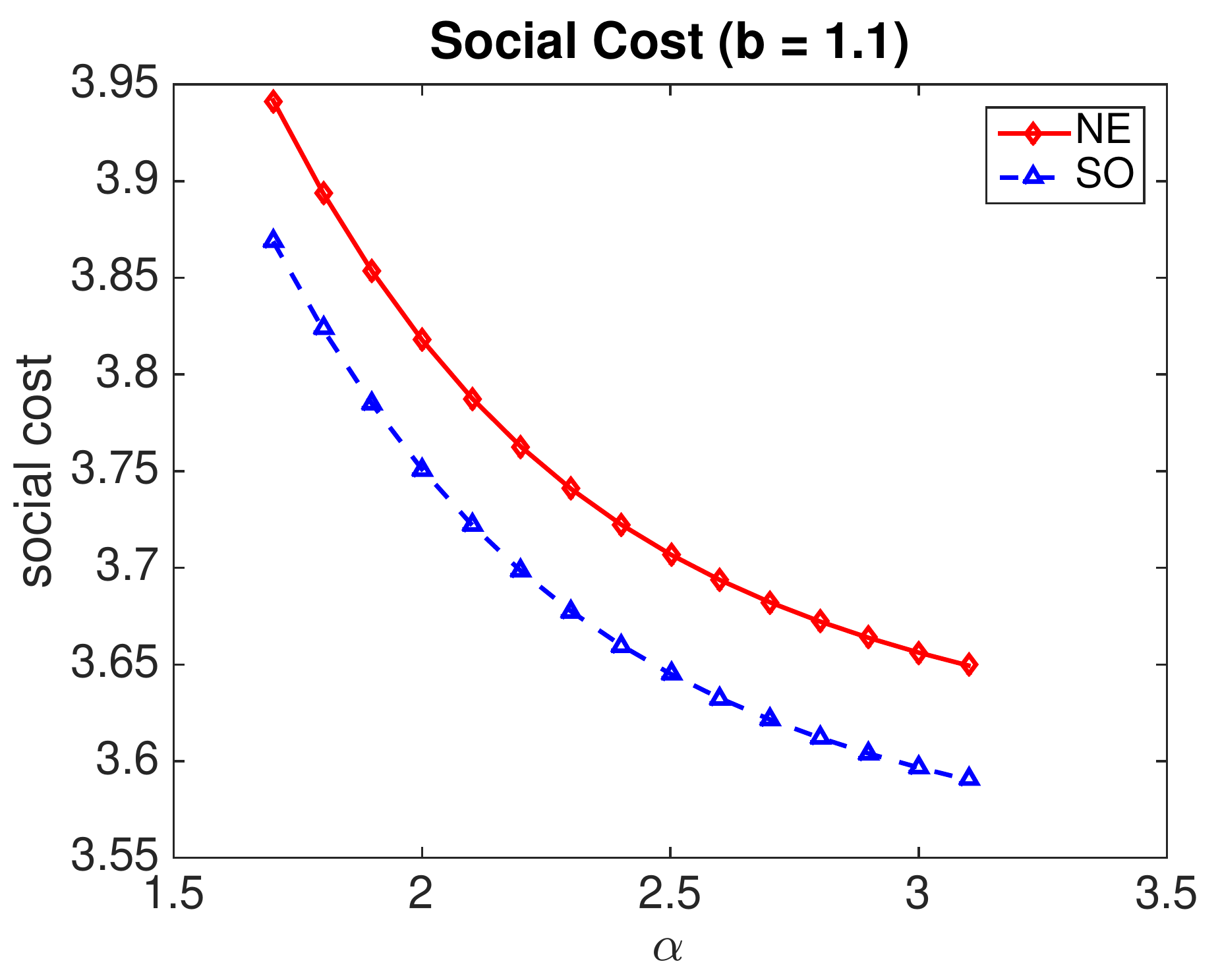}
\includegraphics[width=1.75in]{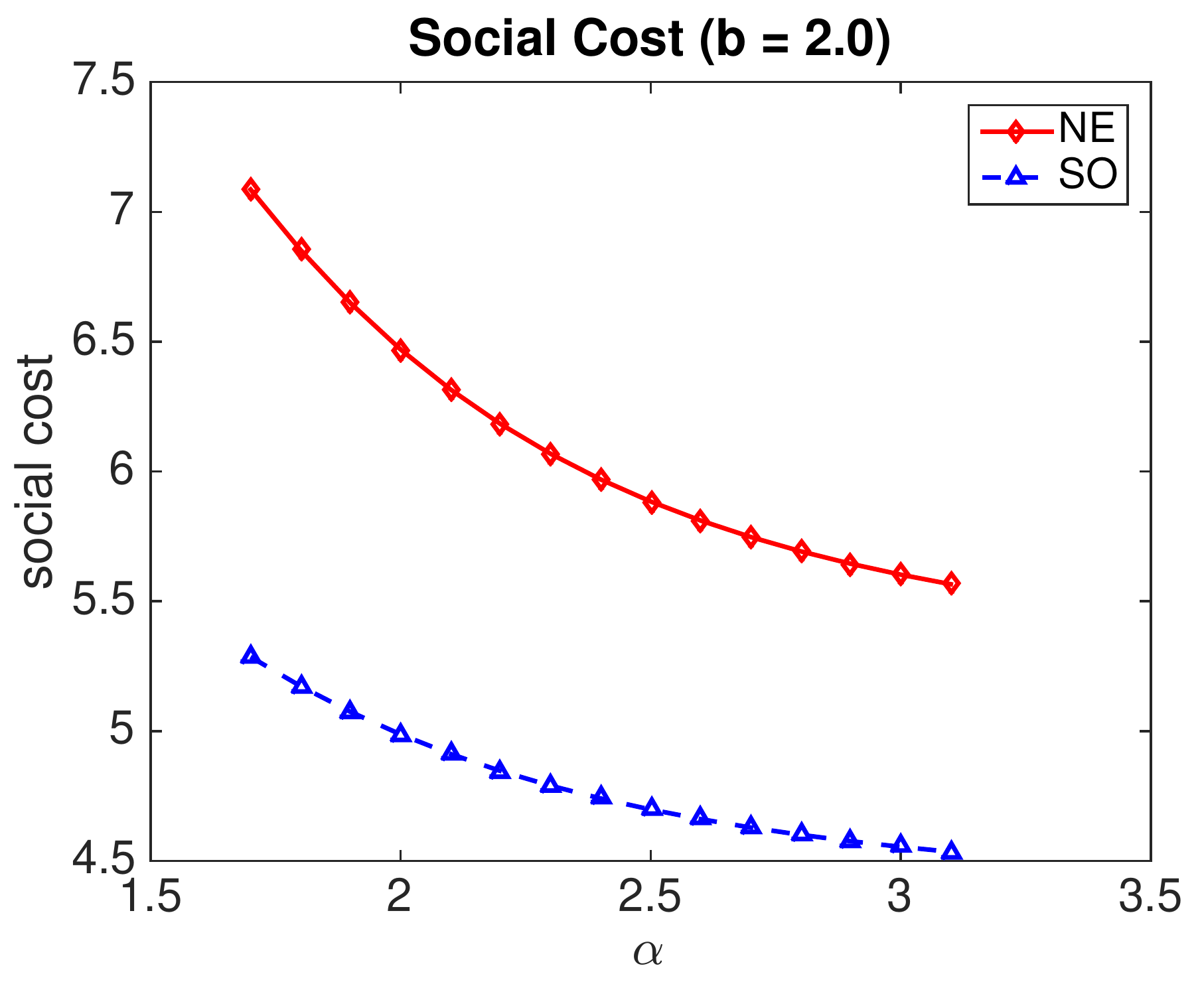}
}
\centerline{(b)}
\centerline{
\includegraphics[width=1.75in]{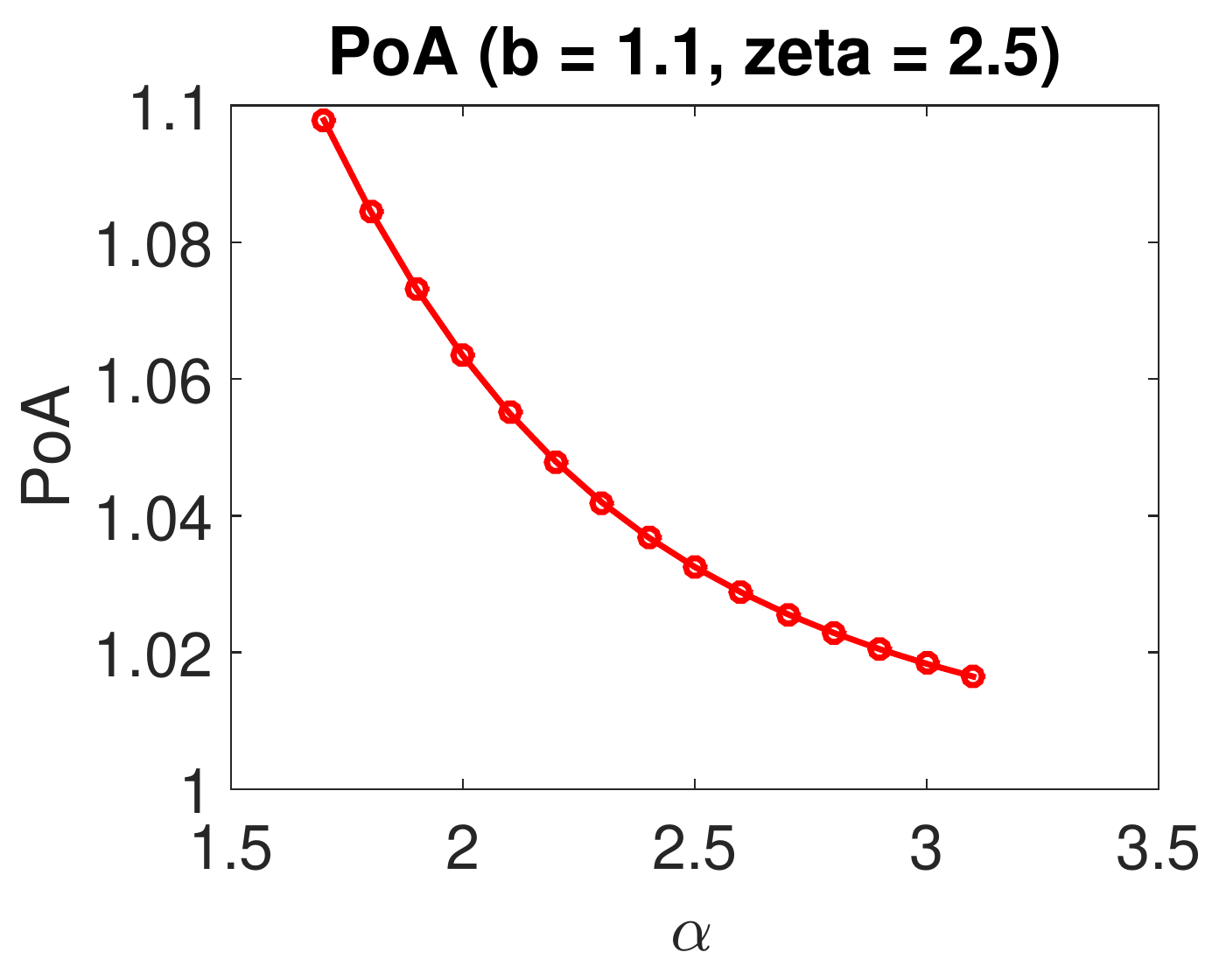}
\includegraphics[width=1.75in]{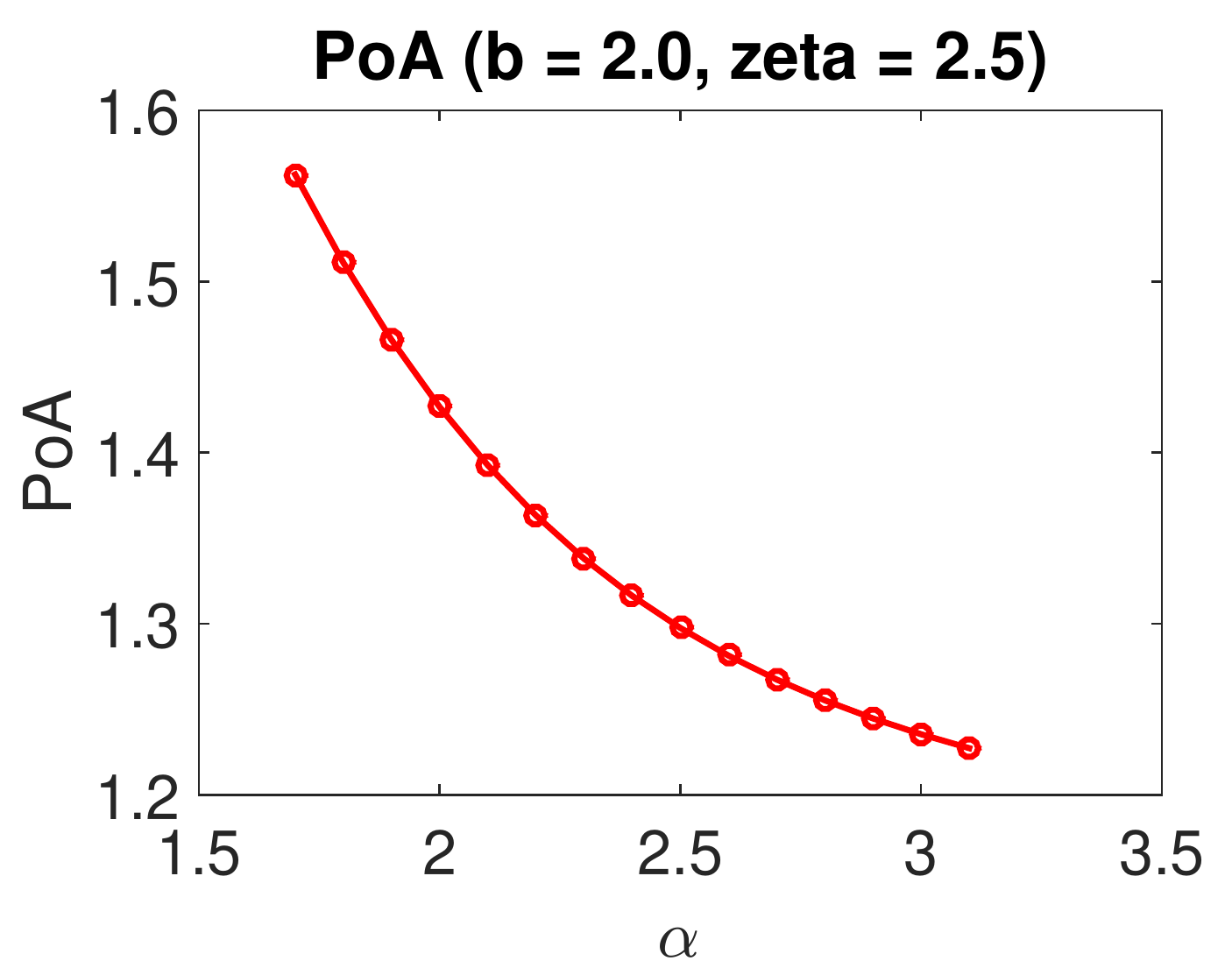}
}
\centerline{(c)}
\caption{(a) Risk exposure, (b) social cost and 
	(c) price of anarchy ($\zeta = 2.5$).}
\label{fig:25}
\end{figure}

In order to examine how the infection probability
function $p$ influences the risk exposure and social
cost, we plot them for $\zeta = 2.5$ (with the 
remaining parameters being identical as before)
in Fig.~\ref{fig:25}. This represents
a scenario in which the available security measures
are more effective in that the infection probability
falls more quickly with increasing security
investments. 

Comparing Figs.~\ref{fig:15} and \ref{fig:25}, we 
observe that although the final risk exposure and
social costs are smaller than in the previous case with 
$\zeta = 1.5$, 
their qualitative behavior remains similar. This 
indicates that as the security measures improve, both 
the risk exposure and the social cost will likely drop. 
However, we suspect that the properties of underlying 
dependence graph will have comparable effects on them.

\section{Conclusions}	\label{sec:Conclusion}

We examined how we could internalize the externalities
produced by security investments of selfish agents in IDS
settings. Our study brought to light an interesting relation
between the local minimizers of social cost and the 
pure-strategy NEs of a related population game in which
the costs of agents are modified to reflect their 
contributions to social cost. Making use of this
relation, we demonstrated that it is possible to reduce
the social cost and enhance the security by imposing 
appropriate penalties on the agents on the basis of the 
losses they suffer as a result of indirect attacks from 
their neighbors. Finally, we proved that as the security 
of agents becomes more interdependent and the weighted 
node degree distribution becomes stochastically
larger, the local security experienced by a node with 
a fixed degree improves.

\begin{appendices}

\section{Proof of Lemma~\ref{lemma:existence}}
	\label{appen:existence}
	
Let $H: \cA^{D_{\max}} \to \cA^{D_{\max}}$, 
where $H_d(\ba) = I^{\opt}(\tau_A + d \cdot e(\ba; \bs))$, 
$d \in \cD$. Then, from Assumption~\ref{assm:pa} and the 
definition of $I^{\opt}$, the mapping $H$ is continuous.
Therefore, since $\cA^{D_{\max}}$ is a compact, convex subset of 
$\R^{D_{\max}}$, the Brouwer's fixed point theorem tells
us that there exists a fixed point of $H$, say $\ba'$, 
such that $H(\ba') = \ba'$. It is clear from the definition 
of a pure-strategy NE in Definition~\ref{defn:PSNE}
that $\ba'$ is a pure-strategy NE.

\section{Proof of Theorem~\ref{thm:unique}}	
	\label{appen:unique}
	
Before we proceed with the proof of theorem, recall that 
$e(\ba; \bs) = g( \gamma_{\avg}(\ba; \bs) )
= g^+( \rho(\ba; \bs) )$, where 
$\gamma_{\avg}(\ba; \bs) = \tau_A \ \beta_{IA} \ \bw(\bs)^T
{\bf p}(\ba) = \tau_A \ \beta_{IA} \ \rho(\ba; \bs)$ and 
$\rho(\ba; \bs) = \bw(\bs)^T {\bf p}(\ba)$. 

In order to prove the theorem, we will first show that if
$\ba^1$ and $\ba^2$ are two pure-strategy NEs, then 
$e(\ba^1; \bs) = e(\ba^2; \bs)$. 
We will prove this claim by contradiction. 
Suppose that the claim is false and the risk exposures are not 
the same. Without loss of 
generality, assume $e(\ba^1; \bs) < e(\ba^2; \bs)$. 
Together with Corollary~\ref{coro:mono}, this 
implies $p(a^1_d) \geq p(a^2_d)$ for all $d \in \cD$ and, 
consequently, 
\beqan
e(\ba^1; \bs)
\myeq g^+ \big( \bw(\bs)^T {\bf p}(\ba^1) \big) \lb
\mygeq g^+ \big( \bw(\bs)^T {\bf p}(\ba^2) \big) 
	= e(\ba^2; \bs),
\eeqan
which contradicts the earlier assumption $e(\ba^1; \bs) < 
e(\ba^2; \bs)$. 

The theorem is now a direct consequence of this claim and
Corollary~\ref{coro:mono}; 
if the theorem is false and there exist two distinct 
pure-strategy NEs $\ba^1$ and $\ba^2$, it necessarily
implies that there is some $d' \in \cD$ such that 
$a^1_{d'} < a^2_{d'}$. However, this is possible only if
$e(\ba^1; \bs) < e(\ba^2; \bs)$ by Corollary
\ref{coro:mono}, which contradicts the
claim proved above.

\section{Proof of Theorem~\ref{thm:relation1}}
	\label{appen:relation1}

First, a local minimizer of social cost, say $\ba^\star$,
must satisfy the 
first-order necessary Karush-Kuhn-Tucker (KKT) condition~\cite{NP}: 
There exist non-negative KKT multipliers $\boldsymbol{\lambda} = 
(\lambda_{d}; \ d \in \cD)$ and   
$\boldsymbol{\mu} = (\mu_{d}; \ d \in \cD )$
such that, for all $d \in \cD$,   
\begin{enumerate}
\item[c1.] $\frac{\partial}{\partial a_d} \tilde{C}_T(\ba^\star, \bs) 
= \lambda_{d} - \mu_{d}$, and

\item[c2.] $\lambda_{d} (a^\star_d - I_{\min}) = 0$ and $\mu_{d} 
(I_{\max} - a^\star_{d}) = 0$. 
\end{enumerate}

From (\ref{eq:SocialCost}) and (\ref{eq:CostPure}), 
\beqan
&& \myhb \frac{\partial}{\partial a_d} \tilde{C}_T(\ba^\star, \bs)
	= \sum_{d' \in \cD} \frac{\partial}{\partial a_d}
	\tilde{C}(\ba^\star, d', \bs) \lb
\myeq \frac{\partial}{\partial a_d} 
		e(\ba^\star; \bs) \sum_{d' \in \cD}
	s_{d'} \ d' \ L(a^\star_{d'}) \lb
&& + s_d \left( \big(\tau_A + d \ e(\ba^\star; \bs) \big) 
	L \ \dot{p}(a^\star_d) + 1 \right),
\eeqan
where
\beqan
\frac{\partial}{\partial a_d} e(\ba^\star; \bs)
\myeq \dot{g}^+\big( \bw(\bs)^T {\bf p}(\ba^\star) \big) \
	w_d(\bs) \ \dot{p}(a^\star_d) \lb
\myeq \dot{g}^+(\rho(\ba^\star; \bs)) \ w_d(\bs) \ \dot{p}(a^\star_d). 
\eeqan
Substituting $w_d(\bs) = d \cdot s_d / d_{\avg}(\bs)$ 
and using $e(\ba; \bs) = g^+\big( \rho(\ba; \bs) \big)$ in 
the KKT conditions, we obtain
\beqa
&& \myhb \frac{\partial}{\partial a_{d}}
\tilde{C}_T(\ba^\star, \bs) \lb
\myeq s_{d} \ d \ \dot{g}^+(\rho(\ba^\star; \bs)) 
	\ \rho(\ba^\star; \bs) \ 
		L \ \dot{p}(a^\star_{d}) \lb
&& + s_{d} \big( ( \tau_A + d \ e(\ba^\star; \bs)) L 
			\ \dot{p}(a^\star_{d}) + 1 \big) \lb
\myeq s_{d} \big( \big( \tau_A + d 
	\ \vartheta(\rho(\ba^\star; \bs) ) \big) 
	 L \ \dot{p}(a^\star_{d}) + 1 \big) \lb
\myeq \lambda_{d} - \mu_{d}.
	\label{eq:2-1}
\eeqa
Dividing both sides by $s_d$ yields
\beqa
\big( \tau_A + d 
	\ \vartheta(\rho(\ba^\star; \bs) ) \big)
	 L \ \dot{p}(a^\star_{d}) + 1 
\myeq \frac{\lambda_d - \mu_d}{s_d}  \lb
\myeq \lambda^n_d - \mu_d^n,  
	\label{eq:4-1}
\eeqa 
where $\lambda^n_d = \lambda_d / s_d$ and 
$\mu^n_d = \mu_d / s_d$. 

At a pure-strategy NE of the modified population game, 
${\ba}^{mc}$, a player in population $d$ faces the constrained 
{\em convex} optimization problem
\beqan
\min_{a \in \cA} \big( \tau_A + d \ \vartheta(\rho({\ba}^{mc}; \bs) ) \big) 
	L(a) + a. 
\eeqan
From the definition of NE, the first-order necessary and 
sufficient KKT condition states that there exist non-negative 
KKT multipliers $\tilde{\lambda}_d$ and $\tilde{\mu}_d$ such 
that
\beqa
\big( \tau_A + d \ \vartheta(\rho(\ba^{mc}; \bs)) \big) L 
	\ \dot{p}({a}^{mc}_d) + 1 
\myeq \tilde{\lambda}_d - \tilde{\mu}_d. 
	\label{eq:4-2}
\eeqa

By comparing (\ref{eq:4-1}) and (\ref{eq:4-2}), we see that  
the local minimizer $\ba^\star$, which satisfies 
the first-order necessary KKT conditions in (\ref{eq:4-1}),
also meets the {\em sufficient} conditions in (\ref{eq:4-2})
with $\tilde{\lambda}_d = \lambda^n_d$ and 
$\tilde{\mu}_d = \mu^n_d$ for all $d \in \cD$. Therefore, it
is a pure-strategy NE of the modified population game.

\section{Proof of Theorem~\ref{thm:GC_NE_Monotonicity}}	
	\label{appen:GC_NE_Monotonicity}

Suppose that the first part of the theorem is false and 
$e(\ba^1; \bs^1) > e(\ba^2; \bs^2)$. By Corollary
\ref{coro:mono}, this implies that, for all $d \in \cD$, 
\beqa
	p(a^1_d) \leq p(a^2_d). 
	\label{eq:thm11-1}
\eeqa
Using $e(\ba; \bs) = g^+\big( \bw(\bs)^T {\bf p}(\ba) 
\big)$ defined in Section~\ref{subsec:PG}, 
\beqa
e(\ba^2; \bs^2)
\myeq g^+\big( \bw(\bs^2)^T {\bf p}(a^2_d) \big) \lb
\mygeq  g^+\big( \bw(\bs^2)^T {\bf p}(a^1_d) \big)
	\label{eq:thm11-2} \\
\mygeq g^+ \big( \bw(\bs^1)^T {\bf p}(a^1_d) \big)
	\label{eq:thm11-3} \\
\myeq e(\ba^1; \bs^1), 
	\nonumber
\eeqa
which contradicts the earlier assumption. The first inequality in
(\ref{eq:thm11-2}) follows from (\ref{eq:thm11-1}),
and the second inequality in (\ref{eq:thm11-3}) is a consequence of
the condition (\ref{eq:thm11}) in Theorem~\ref{thm:GC_NE_Monotonicity} 
and Corollary~\ref{coro:mono}, i.e., $p(a^i_d) \leq p(a^i_{d'})$ if
$d \geq d'$, $i = 1, 2$.

The claim $a^1_d \leq a^2_d$ for all $d \in \cD$ is an immediate 
consequence of the first part and Corollary~\ref{coro:mono}, 
which states that the optimal investment is nondecreasing in the
risk exposure seen by the players. 

The second part of the theorem is true because, under the
stated assumptions that $\bw(\bs^1)$ first-order stochastically
dominates $\bw(\bs^2)$ and $\ba^2 \in$ int($\cA^{D_{\max}}$), 
the inequality (\ref{eq:thm11-3}) is strict as a consequence of
(the strict inequality in) Corollary~\ref{coro:mono}.

\section{Proof of Theorem~\ref{thm:GC_SO_Monotonicity}}	
	\label{appen:GC_SO_Monotonicity}

We prove the first part of the theorem by contradiction. 
Suppose that the claim is false and $e(\ba^\star(\bs^1); \bs^1) 
> e(\ba^\star(\bs^2); \bs^2)$. First, note that 
Corollary~\ref{coro:unique} tells us that 
$\ba^\star(\bs^i)$, $i = 1, 2$, 
are the unique pure-strategy NE of the modified
population games, and $a^\star_d(\bs^i) \leq 
a^\star_{d+1}(\bs^i)$ for all $d = 1, 
\ldots, D_{\max} -1$. Together with this observation, the 
condition (\ref{eq:thm11}) in the theorem tells us that if 
$a^\star_d(\bs^1) \geq a^\star_d(\bs^2)$
for all $d \in \cD$, 
then $\rho(\ba^\star(\bs^1); \bs^1) \leq \rho(\ba^\star(\bs^2); 
\bs^2)$. Consequently,
$e(\ba^\star(\bs^1); \bs^1) \leq e(\ba^\star(\bs^2); \bs^2)$, 
thereby contradicting the assumption $e(\ba^\star(\bs^1); \bs^1) 
> e(\ba^\star(\bs^2); \bs^2)$.
Therefore, there must exist some $d^* \in \cD$ such that 
\beqa
a^\star_{d^*}(\bs^1) < a^\star_{d^*}(\bs^2).
	\label{eq:thm14-1}
\eeqa
We will show that this contradicts the assumption that
$\ba^\star(\bs^1)$ minimizes the social cost. 
For notational simplicity, we denote $\ba^\star(\bs^i)$, 
$i = 1, 2$, by $\ba^i$ throughout the proof. 

The first-order necessary KKT 
conditions tell us that there exist non-negative KKT
multipliers $\pmb{\lambda}^i = (\lambda^i_d ;  
d \in \cD)$ and 
$\pmb{\mu}^i = (\mu^i_d ;  d \in \cD)$, $i = 1, 2$, which  
satisfy
\beqa
&&\myhb \lambda^i_d(a_d^i - I_{\min}) = 0, \ 
\mu^i_d(I_{\max} - a^i_d) = 0, \mbox{ and } \lb
&& \myhb \frac{\partial}{\partial a_d} \tilde{C}_T(\ba^i, \bs^i)
= \lambda^i_d - \mu^i_d, \ \mbox{for all } d \in \cD.
	\label{eq:KKT}
\eeqa

Recall from (\ref{eq:2-1}) in Appendix~\ref{appen:relation1}, 
\beqa
&& \myhb \frac{\partial}{\partial a_{d^*}}
\tilde{C}_T(\ba^1, \bs^1) \lb
\myeq s^1_{d^*} \big( \big( \tau_A + d^* 
	\ \vartheta(\rho(\ba^1; \bs^1))  
		\big) L \ \dot{p}(a^1_{d^*}) + 1 \big) \lb
\myeq \lambda^1_{d^*} - \mu^1_{d^*}.
	\label{eq:3-1}
\eeqa
However, because $a^1_{d^*} < a^2_{d^*} \leq I_{\max}$, we have 
$\mu^1_{d^*} = 0$. Setting $\mu^1_{d^*} = 0$ and dividing both sides 
of (\ref{eq:3-1}) by $s^1_{d^*}$, we get
\beqa
&& \myhb \big( \big( \tau_A + d^* \ 
	\vartheta( \rho(\ba^1; \bs^1)) 
		\big) L \ \dot{p}(a^1_{d^*}) + 1 \big)
= \frac{\lambda^1_{d^*}}{s^1_{d^*}} \geq 0. 
	\label{eq:3-2}
\eeqa

Similarly, since $a^2_{d^*} > a^1_{d^*} \geq 0$, we have 
$\lambda^2_{d^*} = 0$ and 
\beqan
&& \myhb \frac{\partial}{\partial a_{d^*}}
\tilde{C}_T(\ba^2, \bs^2) \lb
\myeq s^2_{d^*} \big( \big( \tau_A + d^* 
	\ \vartheta( \rho(\ba^2; \bs^2) ) \big) 
		L \ \dot{p}(a^2_{d^*}) + 1 \big) \lb
\myeq - \mu^2_{d^*}. 	
\eeqan
Normalizing both sides by $s^2_{d^*}$, 
\beqa
\hspace{-0.1in} \big( \big( \tau_A + d^* \ 
	\vartheta( \rho(\ba^2; \bs^2)) \big) L \ 
		\dot{p}(a^2_{d^*}) + 1 \big) 
= \frac{-\mu^2_{d^*}}{s^2_{d^*}} \leq 0.
	\label{eq:3-3}
\eeqa

Recall that we assumed $e(\ba^1; \bs^1) > 
e(\ba^2 ; \bs^2)$, hence $\rho(\ba^1; \bs^1) 
> \rho(\ba^2; \bs^2)$. In addition, by Assumption 
\ref{assm:pa}, because $a^1_{d^*} < a^2_{d^*}$, 
we have $\dot{p}(a^1_{d^*}) < \dot{p}(a^2_{d^*}) < 0$. 
Therefore, together with Assumption~\ref{assm:vartheta}, 
we have 
\beqan
0 \leq (\ref{eq:3-2}) < (\ref{eq:3-3}) \leq 0
\eeqan
which is a contradiction. 

To prove the second part of the theorem, first note that
the proof of the first part tells us 
$\ba^\star(\bs^1) \geq \ba^\star(\bs^2)$, where the
inequality is elementwise. Second, under 
the stated assumption $\ba^\star(\bs^2) \in$ 
int($\cA^{D_{\max}}$), Corollaries 
\ref{coro:mono} and \ref{coro:unique}
imply $a^\star_{d}(\bs^2) 
< a^\star_{d+1}(\bs^2)$ for all $d = 1, 2, \ldots, 
D_{\max} - 1$. Finally, together with these
observations, the first-order stochastic dominance 
of $\bw(\bs^1)$ over $\bw(\bs^2)$
leads to $e(\ba^\star(\bs^1); \bs^1) 
< e(\ba^\star(\bs^2); \bs^2)$.

\section{Proof of Lemma~\ref{lemma:1}} \label{appen:lemma1}

The following lemma will be used to prove Lemma~\ref{lemma:1}.

\begin{lemma} \label{lemma:2}
Suppose that ${\bf a} = (a_\ell; \ell = 1, \ldots, K)$
and ${\bf b} = (b_\ell;\ell = 1, \ldots, K)$ are two 
finite sequences of nonnegative real numbers of length
$K > 1$ and satisfy 
\beqa
\frac{b_{\ell+1}}{a_{\ell+1}} \leq \frac{b_\ell}{a_\ell}
	\ \mbox{ for all } \ell = 1, \ldots, K-1. 
	\label{eq:lemma2-0}	
\eeqa
Then, 
\beqa
\frac{ \sum_{\ell = 1}^k b_\ell }{ \sum_{\ell = 1}^k a_\ell }
\mygeq \frac{ \sum_{\ell = 1}^{K} b_\ell }
	{ \sum_{\ell = 1}^{K} a_\ell } 
	\ \mbox{ for all } k = 1, \ldots, K.
	\label{eq:lemma2-1}
\eeqa \\ \vspace{-0.1in}
\end{lemma}

Proceeding with the proof of Lemma~\ref{lemma:1}, recall that
the condition (\ref{eq:lemma0}) in Lemma~\ref{lemma:1} states
\beqa
\frac{ (d+1) \ s^2_{d+1} }{ (d+1) \ s^1_{d+1} }
\myleq \frac{ d \cdot s_d^2 }{ d \cdot s_d^1 } 
	\ \mbox{ for all } d \in \cD^-.
	\label{eq:lemma2-4}
\eeqa
From the definition of ${\bf w}_d$, the condition
(\ref{eq:thm11}) is equivalent to 
\beqa
\frac{ \sum_{\ell=1}^d \ell \cdot s^2_\ell  }
	{ \sum_{\ell=1}^d \ell \cdot s^1_\ell  }
\mygeq \frac{ \sum_{\ell=1}^{D_{\max}} \ell \cdot s_\ell^2 }
	{ \sum_{\ell=1}^{D_{\max}} \ell \cdot s_\ell^1 }
	\ \mbox{ for all } d \in \cD.
	\label{eq:lemma2-5}
\eeqa

Let ${\bf a} = (a_d; \ d \in \cD)$ and ${\bf b} = (b_d; \ d \in 
\cD)$, where $a_d = d \cdot s^1_d$ and $b_d = d \cdot s^2_d$. The 
inequalities in (\ref{eq:lemma2-5}) can be rewritten 
in terms of ${\bf a}$ and ${\bf b}$ as 
\beqa
\frac{ \sum_{\ell=1}^d b_\ell }
	{ \sum_{\ell=1}^d a_\ell  }
\mygeq \frac{ \sum_{\ell=1}^{D_{\max}} b_\ell }
	{ \sum_{\ell=1}^{D_{\max}} a_\ell }
	\ \mbox{ for all } d \in \cD.
	\label{eq:lemma2-6}
\eeqa
Moreover, (\ref{eq:lemma2-4}) implies $( b_{d+1} / a_{d+1} )
\leq (b_d / a_d)$, $d \in \cD^-$. The claim of Lemma~\ref{lemma:1} 
in (\ref{eq:lemma2-6}) now follows directly from Lemma~\ref{lemma:2}.

\end{appendices}

\end{document}